\documentclass[12pt]{article}

\usepackage{graphics,graphicx,fullpage,natbib,multirow}
\usepackage{amsmath,amssymb,verbatim,epsfig}
\usepackage[dvipsnames,usenames]{color}
\usepackage{caption,subcaption,booktabs}

\newtheorem{proposition}{Proposition}

\newtheorem{defin}{\bf Definition}

\newenvironment{proof}{\noindent{\bf Proof}}{$\diamond$}

\def\ga{\mbox{Ga}}

\def\be{\mbox{Be}}

\def\bin{\mbox{Bin}}

\def\no{\mbox{N}}
\def\po{\mbox{Po}}

\def\E{\mbox{E}}
\def\V{\mbox{Var}}
\def\Cov{\mbox{Cov}}
\def\Cor{\mbox{Corr}}

\def\P{\mbox{P}}

\def\bep{\mbox{BeP}}

\def\bc{{\bf c}}

\def\bx{{\bf x}}

\def\bZ{{\bf Z}}

\def\simind{\stackrel{\mbox{\scriptsize{ind}}}{\sim}}

\newcommand{\bpi}{\boldsymbol{\pi}}

\newcommand{\btheta}{\boldsymbol{\theta}}
\newcommand{\bupsilon}{\boldsymbol{\upsilon}}

\newcommand{\TT}{\mathbb{T}}

\newcommand{\XX}{\mathcal{X}}

\begin{document}

\baselineskip=24pt

\title{\bf Bayesian nonparametric dynamic hazard rates in evolutionary life tables}
\author{{\sc Luis E. Nieto-Barajas} \\[2mm]
{\sl Department of Statistics, ITAM, Mexico} \\[2mm]
{\small {\tt lnieto@itam.mx}} \\}
\date{}
\maketitle

\begin{abstract}
In the study of life tables the random variable of interest is usually assumed discrete since mortality rates are studied for integer ages. In dynamic life tables a time domain is included to account for the evolution effect of the hazard rates in time. In this article we follow a survival analysis approach and use a nonparametric description of the hazard rates. We construct a discrete time stochastic processes that reflects dependence across age as well as in time. This process is used as a bayesian nonparametric prior distribution for the hazard rates for the study of evolutionary life tables. Prior properties of the process are studied and posterior distributions are derived. We present a simulation study, with the inclusion of right censored observations, as well as a real data analysis to show the performance of our model. 
\end{abstract}

\vspace{0.2in} \noindent {\sl Keywords}: Beta process, discrete time processes, latent variables, moving average process, stationary process.

\section{Introduction}
\label{sec:intro}

Modelling and forecasting mortality rates has been of interest since the last decade of previous century. A mortality rate is basically a hazard rate function of age for a random variable $X$ that indicates the age of death in a population. Non-negative random variables is the subject of study of survival analysis \citep[e.g.][]{klein&moeschberger:03}. Random behaviour is characterised by the survival, $S(x)$, and the hazard, $h(x)$, functions, for a specific age $x>0$. To distinguish the change of mortality in time, we denote by $X_t$ the age of death at time $t$ (year or decade) in a population, for $t\in\TT$, with $\TT$ the set of periods under study. Again, characterization of the set of random variables $\{X_t\}_{t\in\TT}$ can be done by including the time domain in the survival and hazard functions, $S_t(x)$ and $h_t(x)$, respectively. 

Typical models assume that a transformation of the hazard rate function, usually logarithmic or logistic, can be written in terms of age and time effects. For example, \cite{lee&carter:92} consider $\log{h_t(x)}=a_x+b_xk_t+\epsilon_{x,t}$, with $\epsilon_{x,t}$ an error term with zero mean and constant variance. They impose some estimability constraints and use minimum squares to estimate the unknown parameters $a_x$, $b_x$ and $k_t$. Several other authors have further study this model in different directions: disaggregation by sex \citep{lee:00}; improvements in the forecasting step by including two sets of singular value decomposition vectors \citep{renshaw&haberman:03}; including a cohort effect \citep{renshaw&haberman:06}; and modelling of multiple populations via neural networks \citep{richman&wuthrich:19}.

On a different perspective, \cite{currie&al:04} propose a Poisson regression model with P-splines to smooth mortality tables over age and time. \cite{debon&al:08} consider dynamic tables, and after removing trends, model the residuals via a Gaussian random field in age and time. \cite{debon&al:10} summarises several geostatistical approaches and includes cohort effects. 

All previous approaches consider that the hazard rate, $h_t(x)$, has been observed and focus in modelling and smoothing its variability in age and time. However they are ignoring that the ``observed'' $h_t(x)$ is actually a nonparametric estimate of the hazard function, which in discrete age, is defined as $h_t(x)=\P(X_t=x\mid X_t\geq x)$. Therefore a frequentist nonparametric estimate of $h_t(x)$ is the number of people in the population who died at age $x$ and time $t$ divided by the number of people who are still alive at age $x$ for the same time. What we propose is to follow a survival analysis approach and consider each individual a realisation of the random variable $X_t$ and estimate $h_t(x)$ via a bayesian nonparametric approach. For this we introduce a stochastic process that includes dependence across age and time. This process is then used as a nonparametric prior distribution for the hazard rates, which allows us to estimate future behaviours of the hazards in time. An advantage of following a survival analysis approach is that we can include partial information, as censored observations, into the analysis. 

The description of the rest of the paper is as follows: In Section \ref{sec:model} we define the bayesian nonparametric model, sampling and prior distributions, and characterise its properties. Section \ref{sec:post} shows how to obtain posterior inference for the hazard rates under the presence of exact and right censored observations. Section \ref{sec:numerical} considers numerical analysis with simulated and real data. We conclude with some remarks in Section \ref{sec:conclusion}.

\section{Bayesian nonparametric model}
\label{sec:model}

Let $X_t$ be a non-negative discrete random variable, with support $\XX=\{1,\ldots,N_X\}$, that denotes the age of death at time $t\in\TT$, with $\TT=\{1,\ldots,N_T\}$, and $N_X$ and $N_T$ the maximum age and time under study, respectively. The probability law of $X_t$ is described be the density function $f_t(x)=\P(X_t=x)$ for $x\in\XX$, the survival function $S_t(x)=\P(X_t>x)$ for $x>0$, and by the hazard rate function $h_t(x)=\P(X_t=x\mid X_t\geq x)$ for $x\in\XX$, and for $t\in\TT$.

Les us denote the hazard rate for age $x$ at time $t$ by the parameter $\pi_{x,t}$, that is, $h_t(x)=\pi_{x,t}$. Since a discrete hazard rate is a probability, we require that $\pi_{x,t}\in(0,1)$, for all $x\in\XX$ and $t\in\TT$. Then we can re-write the survival and density functions in terms of the parameters $\bpi=\{\pi_{x,t}\}$ as \citep[e.g.][]{klein&moeschberger:03}
\begin{equation}
\label{eq:sf}
S_t(x\mid\bpi)=\prod_{z\leq x}(1-\pi_{z,t})I(x>0)\quad\text{and}\quad f_t(x\mid\bpi)=\pi_{x,t}\prod_{z<x}(1-\pi_{z,t})I(x\in\XX), 
\end{equation}
where  $I(x\in A)$ is the indicator function that takes the value of one if $x\in A$ and zero otherwise. Functions \eqref{eq:sf} can be seen as  nonparametric models defined by $\bpi$.

The idea is to define a convenient nonparametric prior for $\bpi$ that recognises a dependence across different ages $x$'s as well as dependence across different times $t$'s. \cite{nieto&walker:02} proposed a Markov beta process indexed by a single dimension, say $x$, through the use of a latent process. Later on, \cite{jara&al:13} generalised the beta process to allow for dependencies on more than one lag. In both cases, for a particular specification of the parameters, the marginal distribution of the processes are beta. We now extend these ideas to indexes in two dimensions.

Let $\{\pi_{x,t}\}_{(x,t)\in\XX\times\TT}$ be a discrete age-time stochastic process defined for the cross product of sets $\XX$ and $\TT$. For each $(x,t)\in\XX\times\TT$ we require a latent variable $\upsilon_{x,t}$, plus a single latent variable $\omega$, common to all indexes. The idea is to relate the hazards at age $x$ and time $t$ to the hazards of previous ages, say $x-1,\ldots,x-p$, as well as to the hazards of previous years, say $t-1,\ldots,t-q$, where $p$ and $q$ define the order of the dependence. Instead of linking the parameters of interest $\pi_{x,t}$ directly, we link them through the latent $\upsilon_{x,t}$'s. The extra latent parameter $\omega$ plays the role of an anchor needed to obtain a desired marginal distribution. 

In general, we denote by $\partial_{x,t}$ the set of indexes (neighbours) that will influence the location $(x,t)$. Apart from the neighbours, we also need to include the actual location. That is, for an order--$(p,q)$ dependence model, with $p,q>1$, we define $\partial_{x,t}=\{(x,t),(x-1,t),\ldots,(x-p,t),(x,t-1),\ldots,(x,t-q)\}$. To illustrate, in Figure \ref{fig:graph} we show a graphical representation of dependence of orders $p=1$ and $q=1$ for a scenario with $N_X=4$ and $N_T=2$.

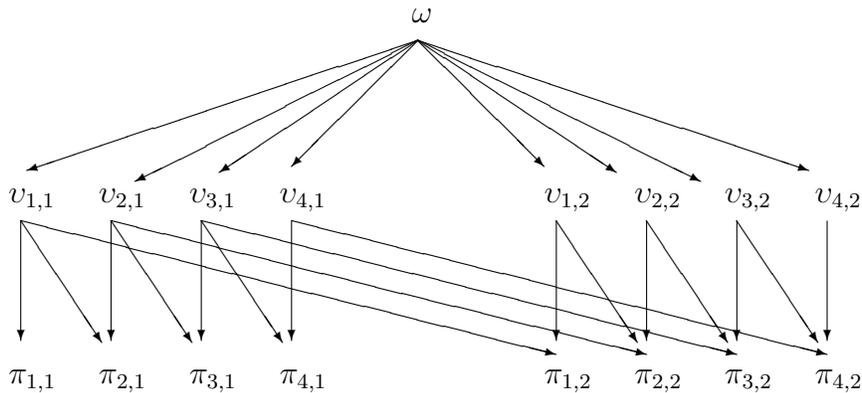
\begin{figure}[ht]
\setlength{\unitlength}{0.8cm}
\begin{center}
\begin{picture}(20,8)
\put(9.5,7.0){$\omega$} 
\put(9.6,6.7){\vector(-3,-1){6.5}}
\put(9.6,6.7){\vector(-2,-1){4.7}}
\put(9.6,6.7){\vector(-3,-2){3.3}}
\put(9.6,6.7){\vector(-1,-1){2.1}}
\put(2.8,4.0){$\upsilon_{1,1}$} 
\put(4.3,4.0){$\upsilon_{2,1}$} 
\put(5.8,4.0){$\upsilon_{3,1}$} 
\put(7.3,4.0){$\upsilon_{4,1}$} 
\put(3.0,3.7){\vector(0,-1){2}}
\put(4.5,3.7){\vector(0,-1){2}}
\put(6.0,3.7){\vector(0,-1){2}}
\put(7.5,3.7){\vector(0,-1){2}}
\put(2.8,1.0){$\pi_{1,1}$} 
\put(4.3,1.0){$\pi_{2,1}$} 
\put(5.8,1.0){$\pi_{3,1}$} 
\put(7.3,1.0){$\pi_{4,1}$} 
\put(3.0,3.7){\vector(2,-3){1.35}}
\put(4.5,3.7){\vector(2,-3){1.35}}
\put(6.0,3.7){\vector(2,-3){1.35}}
\put(9.6,6.7){\vector(3,-1){6.5}}
\put(9.6,6.7){\vector(2,-1){4.7}}
\put(9.6,6.7){\vector(3,-2){3.3}}
\put(9.6,6.7){\vector(1,-1){2.1}}
\put(11.7,4.0){$\upsilon_{1,2}$} 
\put(13.2,4.0){$\upsilon_{2,2}$} 
\put(14.7,4.0){$\upsilon_{3,2}$} 
\put(16.2,4.0){$\upsilon_{4,2}$} 
\put(11.9,3.7){\vector(0,-1){2}}
\put(13.4,3.7){\vector(0,-1){2}}
\put(14.9,3.7){\vector(0,-1){2}}
\put(16.4,3.7){\vector(0,-1){2}}
\put(11.7,1.0){$\pi_{1,2}$} 
\put(13.2,1.0){$\pi_{2,2}$} 
\put(14.7,1.0){$\pi_{3,2}$} 
\put(16.2,1.0){$\pi_{4,2}$} 
\put(11.9,3.7){\vector(2,-3){1.35}}
\put(13.4,3.7){\vector(2,-3){1.35}}
\put(14.9,3.7){\vector(2,-3){1.35}}
\put(3.0,3.7){\vector(4,-1){8.9}}
\put(4.5,3.7){\vector(4,-1){8.9}}
\put(6.0,3.7){\vector(4,-1){8.9}}
\put(7.5,3.7){\vector(4,-1){8.9}}
\end{picture}
\end{center}
\vspace{-0.5cm}
\caption{Graphical representation of dependence of orders $p=1$ and $q=1$.}
\label{fig:graph} 
\end{figure}

Let $\be(a,b)$ denotes a beta distribution with mean $a/(a+b)$, and $\bin(c,\omega)$ denotes a binomial distribution with number of Bernoulli trials $c$ and success probability $\omega$. Let $\bupsilon=\{\upsilon_{x,t}\}$ be the whole set of latent variables, then the law of the process $\{\pi_{x,t}\}$ is defined hierarchically by 
\begin{align}
\nonumber
\omega&\sim\be(a,b)\\
\label{eq:xtbeta}
\upsilon_{x,t}\mid\omega&\simind\bin(c_{x,t},\omega)\\
\nonumber
\pi_{x,t}\mid\bupsilon&\simind\be\left(a+\sum_{(z,s)\in\partial_{x,t}}\upsilon_{z,s}\,,\:b+\sum_{(z,s)\in\partial_{x,t}}(c_{z,s}-\upsilon_{z,s})\right)
\end{align}
where $a$, $b$ and $\bc=\{c_{x,t}\}$ are non-negative parameters for all $x\in\XX$ and $t\in\TT$. We define $\upsilon_{x,t}=0$ with probability one (w.p.1) and $c_{x,t}=0$ for any $x,t<0$. In notation we say $\bpi\sim\bep_{p,q}(a,b,\bc)$, that is, an order--$(p,q)$ beta process with parameters $a$, $b$ and $\bc$, defined by \eqref{eq:xtbeta}. Properties of this beta process are given in Proposition \ref{prop:1}. In particular, the correlation induced and the marginal distribution can be obtained in closed form. 

\begin{proposition}
\label{prop:1}
Let $\{\pi_{x,t}\}_{(x,t)\in\XX\times\TT}$ be a discrete age-time stochastic process $\bep_{p,q}(a,b,\bc)$, defined by equations \eqref{eq:xtbeta}. 
\begin{enumerate} 
\item[(i)] The marginal distribution of $\pi_{x,t}$ for all $(x,t)\in\XX\times\TT$, is $\be(a,b)$, 
\item[(ii)] The correlation between any pair $(\pi_{x,t},\pi_{x',t'})$ is given by $$\Cor(\pi_{x,t},\pi_{x',t'})=\frac{(a+b)\left(\sum_{(z,s)\in\partial_{x,t}\cap\partial_{x',t'}}c_{z,s}\right)+\left(\sum_{(z,s)\in\partial_{x,t}}c_{z,s}\right)\left(\sum_{(z,s)\in\partial_{x',t'}}c_{z,s}\right)}{\left(a+b+\sum_{(z,s)\in\partial_{x,t}}c_{z,s}\right)\left(a+b+\sum_{(z,s)\in\partial_{x',t'}}c_{z,s}\right)},$$
\item[(iii)] If $c_{x,t}=0$ for all $(x,t)\in\XX\times\TT$ then the $\pi_{x,t}$'s become independent. 
\end{enumerate}
\end{proposition}
\begin{proof}
For $(i)$ we note that conditionally on $\omega$ the $\upsilon_{x,t}$ are independent binomials, so the sum is again binomial, i.e., $\sum_{(z,s)\in\partial_{x,t}}\upsilon_{z,s}\mid\omega\sim\bin\left(\sum_{(z,s)\in\partial_{x,t}}c_{z,s},\omega\right)$. Integrating out $\omega$, the marginal distribution of $\sum_{(z,s)\in\partial_{x,t}}\upsilon_{z,s}$ is beta-binomial with parameters $(a,b,\sum_{(z,s)\in\partial_{x,t}}c_{z,s})$. Using conjugacy properties between the beta and the binomial distributions \citep[e.g.][]{bernardo&smith:00} we obtain $\pi_{x,t}\sim\be(a,b)$ marginally. For $(ii)$ we rely on conditional independence properties and the iterative covariance formula. Then $\Cov(\pi_{x,t},\pi_{x',t'})=\E\{\Cov(\pi_{x,t},\pi_{x',t'}\mid\bupsilon)\}+\Cov\{\E(\pi_{x,t}\mid\bupsilon),\E(\pi_{x',t'}\mid\bupsilon)\}$. The first term in the sum becomes zero since $\pi_{x,t}$'s are conditional independent given $\bupsilon$. The second term, after removing the additive constants of the expected values, is rewritten as $\Cov\left(\sum_{(z,s)\in\partial_{x,t}}\upsilon_{z,s},\sum_{(z,s)\in\partial_{x',t'}}\upsilon_{z,s}\right)$ divided by $(a+b+\sum_{(z,s)\in\partial_{x,t}}c_{z,s})(a+b+\sum_{(z,s)\in\partial_{x',t'}}c_{z,s})$. Concentrating in the numerator, and using the iterative covariance formula for a second time, we get $\E\{\Cov(\sum_{(z,s)\in\partial_{x,t}}\upsilon_{z,s},\sum_{(z,s)\in\partial_{x',t'}}\upsilon_{z,s}\mid\omega)\}+\Cov\{\E(\sum_{(z,s)\in\partial_{x,t}}\upsilon_{z,s}\mid\omega),$ \linebreak[4] $\E(\sum_{(z,s)\in\partial_{x',t'}}\upsilon_{z,s}\mid\omega)\}$. The first term, after separating the sums in the common part, reduces to $\E\{\V(\sum_{(z,s)\in\partial_{x,t}\cap\partial_{x',t'}}\upsilon_{z,s}\mid\omega)\}$ which becomes $(\sum_{(z,s)\in\partial_{x,t}\cap\partial_{x',t'}}c_{z,s})\E\{\omega(1-\omega)\}$. After computing the expected value, this can be re-written as $(\sum_{(z,s)\in\partial_{x,t}\cap\partial_{x',t'}}c_{z,s})(a+b)\V(\omega)$. The second term, after computing the expectations becomes \linebreak[4] $(\sum_{(z,s)\in\partial_{x,t}}c_{z,s})(\sum_{(z,s)\in\partial_{x',t'}}c_{z,s})\V(\omega)$. Finally, using $(i)$, we note that $\V(\pi_{x,t})=\V(\pi_{x',t'})$ $=\V(\omega)$, so dividing the covariance between the product of the standard deviations we obtain the result. For $(iii)$ we note that if $c_{x,t}=0$ then $\upsilon_{x,t}=0$ w.p.1. If we do this for all $(x,t)\in\XX,\TT$ then it is straightforward to see that the $\pi_{x,t}'s$ become independent.
\end{proof}

\bigskip
Many things can be concluded from Proposition \ref{prop:1}. Parameters $a$ and $b$ are the two shape parameters of the marginal beta distribution of the process. Parameters $\bc$ control the strength of the correlation between any two elements $\pi_{x,t}$ and $\pi_{x',t'}$. The correlation is stronger, if the two elements share more parameters. Some locations can be more influential than others if their corresponding $c_{x,t}$ parameter is larger. Moreover, the process $\{\pi_{x,t}\}$ becomes strictly stationary if the parameters $\bc$ are all equal. That is, if $c_{x,t}=c$ for all $x$ and $t$, then the correlation for $x,x'>p$ and $t,t'>q$ simplifies to $$\Cor(\pi_{x,t},\pi_{x',t'})=\frac{(a+b)\{\#(\partial_{x,t}\cap\partial_{x',t'})\}c+(p+q+1)^2c^2}{\{a+b+(p+q+1)c\}^2},$$
where $\#(\partial_{x,t}\cap\partial_{x',t'})$ denotes the number of elements in the intersection set. 

Therefore, our bayesian nonparametric model is defined by the nonparametric sampling distribution \eqref{eq:sf} together with process $\bep_{p,q}(a,b,\bc)$, defined by \eqref{eq:xtbeta}, as nonparametric prior for $\bpi$. We now proceed to show how posterior inference is done with this model.

\section{Posterior analysis}
\label{sec:post}

Let us assume that, for each year $t\in\TT$, we observe a sample of size $n_t$, that is, $X_{1,t},X_{2,t},\ldots,$ $X_{n_t,t}$, with possible right censored observations. Let $\delta_{i,t}$ be the censored indicator that takes the value of $1$ if the observation is exact and zero if it is right-censored. Then, the joint distribution of the sample is $$f(\bx\mid\pi)=\prod_{t\in\TT}\prod_{i=1}^{n_t}f(x_{i,t}\mid\bpi)^{\delta_{i,t}}S(x_{i,t}\mid\bpi)^{1-\delta_{i,t}}.$$

Using expressions \eqref{eq:sf}, we can substitute the survival and density functions in terms of $\bpi$ to obtain 
$$f(\bx\mid\bpi)=\prod_{t\in\TT}\prod_{x\in\XX}\pi_{x,t}^{r_{x,t}}(1-\pi_{x,t})^{m_{x,t}-r_{x,t}},$$
where $r_{x,t}=\sum_{i=1}^{n_t}I(X_{i,t}=x,\,\delta_{i,t}=1)$ is the number of people dying and $m_{x,t}=\sum_{i=1}^{n_t}I(X_{i,t}\geq x)$ is the number of people at risk, at age $x$ and time $t$, respectively.

The prior distribution for $\bpi$ is characterised by \eqref{eq:xtbeta} and its marginal distribution could be obtained via integration with respect to the latent variables. However, inference is simpler if we consider the joint prior distribution for the whole vector that includes the latent variables $\bupsilon$ and $\omega$. This has the form $$\hspace{-6cm}f(\bpi,\bupsilon,\omega)=\be(\omega\mid a,b)\prod_{x\in\XX}\prod_{t\in\TT}\bin(\upsilon_{x,t}\mid c_{x,t},\omega)$$ $$\hspace{6cm}\times\be\left(\pi_{x,t}\left|a+\sum_{(z,s)\in\partial_{x,t}}\upsilon_{z,s}\,,\:b+\sum_{(z,s)\in\partial_{x,t}}(c_{z,s}-\upsilon_{z,s})\right.\right),$$
where the distributions with an argument upfront denote the corresponding density function. 

The posterior distribution for the extended vector of parameters is simply $f(\bpi,\bupsilon,\omega\mid\bx)\propto f(x\mid\bpi)f(\bpi,\bupsilon,\omega)$. To obtain summaries from it we suggest to implement a Gibbs sampler \citep{smith&roberts:93} with the following posterior conditional distributions.

\begin{enumerate}
\item[i.] 
Conditional distribution of the latent $\omega$ 
$$f(\omega|\bupsilon)=\be\left(\omega\left|\,a+\sum_{x\in\XX}\sum_{t\in\TT}\upsilon_{x,t}\,,\,b+\sum_{x\in\XX}\sum_{t\in\TT}(c_{x,t}-\upsilon_{x,t})\right.\right).$$
\item[ii.]
Conditional distribution of the latent variables $\upsilon_{x,t}$, for $x\in\XX$ and $t\in\TT$
$$f(\upsilon_{x,t}|\bpi,\omega)\propto\frac{{c_{x,t}\choose\upsilon_{x,t}}\left\{\left(\frac{\omega}{1-\omega}\right)\prod_{(z,s)\in\varrho_{x,t}}\left(\frac{\pi_{z,s}}{1-\pi_{z,s}}\right)\right\}^{\upsilon_{x,t}}I(\upsilon_{x,t}\in\{0,1,\ldots,c_{x,t}\})}{\prod_{(z,s)\in\varrho_{x,t}}\Gamma\left(a+\sum_{(y,u)\in\partial_{z,s}}\upsilon_{y,u}\right)\Gamma\left(b+\sum_{(y,u)\in\partial_{z,s}}\left(c_{y,u}-\upsilon_{y,u}\right)\right)}.$$ 
\item[iii.]
Conditional distribution of the parameters $\pi_{x,t}$, for $x\in\XX$ and $t\in\TT$
$$f(\pi_{x,t}\mid\bupsilon,\bx)=\be\left(\pi_{x,t}\left|\,a+\sum_{(z,s)\in\partial_{x,t}}\upsilon_{z,s}+r_{x,t}\,,\:b+\sum_{(z,s)\in\partial_{x,t}}(c_{z,s}-\upsilon_{z,s})+m_{x,t}-r_{x,t}\right.\right).$$
\end{enumerate}

Sampling from conditional distributions (i) and (iii) is straightforward, since they have standard form (a beta distribution). Sampling from (ii) is not complicated due to the discreteness of the density in a bounded support, so it can be evaluated and normalised. 

Although we do not have a complete analytical expression for the posterior distribution of the hazard rate parameters $\pi_{x,t}$, we can have an idea how the data and the prior dependence come to play. Taking iterative expectation from (iii) we have that the posterior mean of $\pi_{x,t}$ has the form 
\begin{equation}
\label{eq:postm}
\E(\pi_{x,t}\mid\bx)=\frac{a+\sum_{(z,s)\in\partial_{x,t}}\E(\upsilon_{z,s}\mid\bx)+r_{x,t}}{a+b+\sum_{(z,s)\in\partial_{x,t}}c_{z,s}+m_{x,t}}. 
\end{equation}
Recalling that if $c_{x,t}=0$ for all $x$ and $t$ this implies that $\upsilon_{x,t}=0$ w.p.1, so regardless of the data, $\E(\upsilon_{x,t}\mid\bx)=0$. If we further take $a$ and $b$ close to zero, which would correspond to a vague (non-informative) prior, then \eqref{eq:postm} reduces to the frequentist non-parametric estimate of the hazard rate $r_{x,t}/m_{x,t}$. Here ``vague'' is interpreted as a prior with large variance, however, other non-informative priors include Jeffrey's prior ($a=b=0.5$) or the uniform prior ($a=b=1$) \citep[e.g.][]{berger:85}. For any $c_{x,t}>0$, we know that $0\leq\E(\upsilon_{x,t}\mid\bx)\leq c_{x,t}$, so $\sum_{(z,s)\in\partial_{x,t}}\E(\upsilon_{z,s}\mid\bx)$ in \eqref{eq:postm} can be seen as extra people dying and $\sum_{(z,s)\in\partial_{x,t}}c_{z,s}$ would be additional people at risk, at age $x$ and time $t$. So ages-times $(x,t)$ more influential in their neighbours would have larger $\upsilon_{x,t}$ values.

\section{Numerical studies}
\label{sec:numerical}

\subsection{Simulation study}

In this first study we test our model in a controlled scenario with the presence of censoring. Let $Y_t\sim\po(\lambda_t)$ be a poisson random variable with mean $\lambda_t$. We define $Z_t=Y_t+1$ such that $Z_t$ is a strictly positive random variable, and denote it as $Z_t\sim\po^+(\lambda_t)$. Let $C_t$ be a censoring age such that $C_t\sim\po^+(\lambda_C)$. We simulated data points $Z_{i,t}$ and censoring ages $C_{i,t}$ independently, and define observations $X_{i,t}=\min(Z_{i,t},C_{i,t})$ and right-censoring indicators $\delta_{i,t}=I(Z_{i,t}\leq C_{i,t})$, for $i=1,\ldots,n_t$ and $t=1,\ldots,N_T$. 

We took $\lambda_t=8+(t-1)/2$, which implies that the hazard rate $h_t(x)$ has an increasing trend in age $x$, for a fixed time $t$, and a decreasing trend in time $t$, for a fixed age $x$. Figure \ref{fig:trueh} shows these hazard rates, where the intensity of the colour increases for larger times $t$. Sample sizes are $n_t=1,000$ for $t=1,\ldots,N_T$ and $N_T=15$. For the censoring rate we took $\lambda_C=18$, and since this is fixed across time $t$, the percentage of censoring values increases for larger $t$. Specifically, they go from very light censoring ($2.89\%$) for $t=1$ to moderate censoring ($33.08\%$) for $t=15$. Given the possible simulated values, and to avoid numerical problems we took a maximum age $N_X=18$ to define the hazards $\pi_{x,t}$. To study the prediction properties of our model, we decided to estimate the hazard rates for out of sample times $t=16,17$.

To specify the prior distribution for $\pi_{x,t}$ we took $a=b=0.001$ and assumed $c_{x,t}=c$ with a range of values for $c\in\{2,5,10,20\}$ to compare. For the order of dependence we took $p=1,\ldots,10$ and $q=1,\ldots,10$ to explore. These values imply a total of $4\times 10\times 10=400$ different model specifications. For each of them we implemented a Gibbs sampler with conditional distributions $(i)$, $(ii)$ and $(iii)$, and ran it for $18,000$ iterations with a burn in of $6,000$ and keeping one of every $3^{rd}$ iteration. Convergence of the chains was assessed informally by looking at the trace plots, ergodic means and autocorrelation functions. Figure \ref{fig:mcmc} shows these convergence diagnostics for parameter $\omega$ when $p=10$, $q=8$ and $c=5$. 

To assess model fit we computed the L-measure which is a goodness of fit statistic that summarises variance and mean square error (bias) of the posterior distribution of each $\pi_{x,t}$. This is defined as \citep{laud&ibrahim:95}
\begin{equation}
\label{eq:lmeasure}
L(\nu)=\frac{1}{N_X N_T}\sum_{x=1}^{N_X}\sum_{t=1}^{N_T}\V\left(\pi_{x,t}\mid\bx\right)+\frac{\nu}{N_X N_T}\sum_{x=1}^{N_X}\sum_{t=1}^{N_T}\left\{\E\left(\pi_{x,t}\mid\bx\right)-\pi_{x,t}^{0}\right\}^2,
\end{equation}
where $\pi_{x,t}^0$ denotes the true hazard from the positive poisson model $\po^+(\lambda_t)$, which does not have close analytical expression but can be evaluated numerically, and $\nu\in[0,1]$ is a constant that determines the importance of the bias in the overall measure. To assess forecasting performance we computed the L-measure \eqref{eq:lmeasure} but considered, in the second summation, the range of values $t=N_t+1,\ldots,N_T+N$, which in this case $N=2$. Note that $\pi_{x,t}^0$ is also available for this range of values of $t$.

Figure \ref{fig:lmeaS} shows the values of the L-measure, with $\nu=1/2$, for the different prior specifications and for the observed range of times (in sample) and predictions (out of sample). In sample L-measure decreases for larger $p$ and larger $q$, with little gain for $q\geq 8$. The measure also takes smaller values for $c=5$. Therefore, we select our best fitting model that with $p=10$, $q=8$ and $c=5$. To place our model in context, the L-measure for the model that ignores dependence, obtained with $c_{x,t}=0$ for all $x$ and $t$, is 0.0023, which is a lot higher than the values for most of our models. On the other hand, the out of sample L-measure (bottom panel) prefers smaller values of $p$, larger values of $q$ and larger values of $c$. In general there is little gain for values $c\geq 10$. Therefore our best forecasting model is that with $p=1$, $q=10$ and $c=10$. 

We summarise posterior distributions of the hazard rates $\pi_{x,t}$ with the mean, as a point estimate, together with $2.5\%$ and $97.5\%$ quantiles, to define $95\%$ credible intervals. Figure \ref{fig:posth} includes posterior estimates for the first two times $t=1,2$ and for the out of sample times $t=16,17$. For each time we show two estimates, one with independent priors obtained by setting $c=0$ (left panel), and another with our best model (right panel). For all times, $95\%$ credible intervals are a lot wider with independent priors than with our best dependent model. Posterior estimates obtained with our best models correct the lack of data by producing more precise estimates and are able to accurately forecast the hazard rates for future times.

\subsection{Real data analysis}

The human mortality database \citep{hmd:21} is a project conducted by the Department of Demography at the University of California in Berkeley, USA, and by the Max Planck Institute for Demographic Research in Rostock, Germany. They provide death counts and population sizes for different countries needed to construct life tables. One of the most complete datasets is that of Switzerland, which contains records from 1876 up to 2018. We decided to constrain the years of study to the window from 1925 to 2014 and formed 18 quinquennial groups: $[1925,1929],[1930,1934],\ldots,[2010,2014]$. The ages of death were also constrained from 0 to 99 and formed 20 quinquennial groups: $[0,4],[5,9],\ldots,[95,99]$. In both types of groups the deaths and the population were aggregated. The only consequence of doing this is that estimated hazard rates for the quinquennial group would be an average of the hazard rates for individual ages/years. 

In notation of Section \ref{sec:post}, available data are death counts at age $x$ and year $t$, which correspond to $r_{x,t}$ in a scenario with no censoring, and population sizes at age $x$ and year $t$, $Q_{x,t}$. We recover the number of people at risk by $m_{x,t}=\sum_{y=x}^{N_X}Q_{y,t}$. Frequentist nonparametric estimates of the hazard rates are $\widehat\pi_{x,t}=r_{x,t}/m_{x,t}$, for $x=1,\ldots,N_X$ and $t=1,\ldots,N_T$ with $N_X=20$ and $N_T=18$. These are shown in Figure \ref{fig:swiss0}, where the intensity of the colour increases for larger $t$ (more recent dates). As we can see, the mortality rates are monotonically increasing along age, and shown a slow but steady decreasing path along time. The latter is a consequence of a larger life expectancy in developed countries like Switzerland. The small gap that can be appreciated in the last ages occurs in the 70's, that is, between groups $[1970,1974]$ and $[1975,1979]$. 

Having a closer look at the posterior conditional distribution (iii) in Section \ref{sec:post}, we note that observed data, summarised in sufficient statistics $r_{x,t}$ and $m_{x,t}$, enter the expression additively to the latent variables $\upsilon_{x,t}$ and $c_{x,t}-\upsilon_{x,t}$. Therefore, to define the prior distributions, the values of $c_{x,t}$ have to be of the same order of $r_{x,t}$. Moreover, to avoid numerical problems we decided to divide both, the number of deaths and the population sizes, by a constant $k=10$, which has no implications in the hazard rate estimates.

We fitted our bayesian nonparametric dependent model with the following prior specifications: $a=b=0.001$ and assumed $c_{x,t}=c$ with a range of values for $c\in\{5,20,50,100\}$ to compare. Due the point just raised,  these chosen values for $c$ are a lot higher than those considered in the simulation study. For the order of dependence we took $p=1,2,3$ and $q=1,\ldots,10$ to explore.

We fitted our model for the Switzerland dataset for $t=1,\ldots,17$ and decided to leave out the last quinquennial time group, $t=18$, for out of sample comparison. We implemented a Gibbs sampler with $10,000$ iterations with a burn in of $1,000$ and keeping one of every second iteration. Again, convergence of the chain was assessed informally with good convergence behaviour for all chains. As a goodness of fit for the different prior choices we also used the L-measure but computed it for two cases: using in sample estimations for $t=1,\ldots,17$; and out of sample estimations for $t=18$. Since we do not know the true hazard rates $\pi_{x,t}^0$, needed in \eqref{eq:lmeasure}, we use the frequentist nonparametric estimates $\widehat\pi_{x,t}$ instead. 

Figure \ref{fig:lmeaR} summarises both L-measures with $\nu=1/2$. For both cases, in sample and out of sample, smaller values of $p$ and larger values of $q$ are preferred. On the other hand, there is an opposite preference in the parameter $c$. For in sample estimation, smaller $c$ is preferred, whereas for out of sample larger values of $c$ are preferred. Looking closer to the out of sample L-measures, they are practically the same for $c$ values in $\{20,50,100\}$ for $p\geq 6$. Making a compromise between the two sets of measures, we take as our best model that with $p=1$, $q=10$ and $c=20$. 

Finally, in Figure \ref{fig:swiss1} we show posterior estimates for some times. Again, we take the posterior mean and quantiles $2.5\%$ and $97.5\%$ of $\pi_{x,t}$ to define point and $95\%$ credible interval estimates. In the top row we show estimates for $t=1$ and compare independence prior (left) with our best dependent model fit (right). Little difference can be seen in both fits, perhaps a smaller variance in the dependence model for ages $x\ge 19$. In the bottom row we show out of sample estimates for $t=18$ and compare estimates with $p=1$, $q=10$ and two values of $c$, $c=5$ (left) and $c=20$ (right). We can appreciate the reduction in the length of the credible intervals when increasing the value of $c$ from $5$ to $20$, however in both cases the out of sample estimations are reasonably good.  

To place our model in context, we also fitted \cite{lee&carter:92}'s model (LC) described in Section \ref{sec:intro}. For the errors we assumed $\epsilon_{x,t}\sim\no(0,\tau_\epsilon)$, that is a zero mean normal distribution with precision $\tau_\epsilon$. We performed a bayesian analysis and took independent prior distributions $a_x\sim\no(0,0.1)$,  $b_x\sim\no(0,0.1)$, and a dynamic linear prior $k_t\sim\no(k_{t-1},\tau_k)$ for $t=2,\ldots,N_T$ with $k_1\sim\no(0,\tau_k)$. Finally for the precisions we took $\tau_\epsilon\sim\ga(0.1,0.1)$ and $\tau_k\sim\ga(0.1,0.1)$. We also assessed forecasting performance of the models by considering the following scenario: Fitted the model for $t=1,\ldots,N_T-N$ and estimated $N$ future hazard rates for $t=N_T-N+1,\ldots,N_T$; we took values $N=1,\ldots,10$. Posterior inference of both models was obtained by running a Gibbs sampler with the same specifications as above. We did our own implementation of LC model in jags \citep{plummer:03}. 

The L-measures, in sample and out of sample, for our model with $p=1$, $q=10$ and $c=20$, and for the LC model are presented in Table \ref{tab:gof}. The in sample L-measures are somehow stable in both models for the different values of $N$, with an slight increase for our model and a slight decrease for the LC model, for larger $N$. In general, the fitting of our model is better than LC model. For comparing the forecasting performance (out of sample L-measures) between the two models, we have reported the ratio of the two L-measures (R$_{\rm out}$) in the last column of Table \ref{tab:gof}. Here something interesting happens. LC model is better than our model for $N\leq 4$, but for $N>4$ our model shows better forecasting behaviour than the LC model. 

In summary, our model has shown a better fitting than the LC model and our model outperforms LC model in long term forecasting, for the particular dataset analysed here.

\section{Concluding Remarks}
\label{sec:conclusion}

We have proposed a bayesian nonparametric model for the study of dynamic life tables. Prior distributions for the hazard rates is based on a stochastic process that allows for dependence in the hazards across ages and times. The order of dependence is controlled by parameters $p$ and $q$ in each dimension, respectively, and the strength of dependence is controlled by a set of parameters $\bc=\{c_{x,t}\}$. In contrast to other proposals, our model admits the presence of right censored observations, which are not uncommon in life tables estimation. 

Posterior inference in our model relies in the implementation of a Gibbs sampler, which for this paper was implemented in Fortran. Computational times do not increase with sample size, they only depend on the number of ages, $N_X$, and the number of times, $N_T$, under study. For the sizes considered in the numerical examples and for $18,000$ iterations the code takes less than 2 minutes to run. 

Our model is unique in the sense that it can not only estimate future hazard functions, but it can improve hazard rates estimation for observed ages and times, specially in the ages where few deaths are observed or where censored observations are present. Our model forecasting performance improves for longer times ahead. 

The inclusion of covariates to define sub-populations, like sex, states or countries, might also be of interest. One way of dealing with this is to consider a multiplicative hazard model, similar to \cite{cox:72}, but ensuring that the effect of covariates still defines a proper discrete hazard rate. If we denote by $\bZ_{it}$ the covariate vector of individual $i$ and by $\btheta$ a coefficients vector, then the hazard function for individual $i$ could be $h_{it}(x)=\varphi(\bZ_{it},\btheta)h_{0t}(x)$, with $\varphi(\bZ_{it},\btheta)=\exp(\bZ_{it}'\btheta)/\{1+\exp(\bZ_{it}'\btheta)\}$ and $h_{0t}(x)$ a discrete baseline hazard. Furthermore, the stochastic process $\bep_{p,q}(a,b,\bc)$, as in \eqref{eq:xtbeta}, could be used as a nonparametric prior distribution for the baseline hazard. Studying this and other alternative models to cope with covariates is left to study for a future work.

\section*{Acknowledgements}

The author acknowledges support from \textit{Asociaci\'on Mexicana de Cultura, A.C.} and is grateful to valuable comments from the reviewers.

\bibliographystyle{natbib}

\begin{thebibliography}{99}

\bibitem[Berger, 1985]{berger:85}
Berger, J.O. (1985). \textit{Statistical Decision Theory and Bayesian Analysis}. Springer: New York.

\bibitem[Bernardo and Smith, 2000]{bernardo&smith:00}
Bernardo, J.M. and Smith, A.F.M. (2000). \textit{Bayesian Theory}. Wiley: New York.

\bibitem[Currie et al., 2004]{currie&al:04}
Currie, I.D., Durban, M. and Eilers, P.H.C. (2004). Smoothing and forecasting mortality rates. \textit{Statistical Modelling} \textbf{4}, 279--298.

\bibitem[Cox, 1972]{cox:72}
Cox, D.R. (1972). Regression models and life tables. \textit{Journal of the Royal Statistical Society, Series B} \textbf{34}, 187--220.

\bibitem[Deb\'on et al., 2008]{debon&al:08}
Deb\'on, A., Montes, F., Mateu, J., Porcu, E. and Bevilacqua, M. (2008). Modelling residuals dependence in dynamic life tables: A geostatistical approach. \textit{Computational Statistics and Data Analysis} \textbf{52}, 3128--3147.

\bibitem[Deb\'on et al., 2010]{debon&al:10}
Deb\'on, A., Mart\'inez-Ruiz, F. and Montes, F. (2010). A geostatistical approach for dynamic life tables: The effect of mortality on remaining lifetime and annuities. \textit{Insurance: Mathematics and Economics} \textbf{47}, 327--336.

\bibitem[HMD, 2021]{hmd:21}
Human Mortality Database. University of California, Berkeley (USA), and Max Planck Institute for Demographic Research (Germany). Available at www.mortality.org or www.humanmortality.de (data downloaded on February 11, 2021).

\bibitem[Laud and Ibrahim, 1995]{laud&ibrahim:95}
Laud, P. and Ibrahim, J. (1995). Predictive model selection. \textit{Journal of the Royal Statistical Society, Series B} \textbf{57}, 247--262.

\bibitem[Jara et al., 2013]{jara&al:13}
Jara, A., Nieto-Barajas, L.E. and Quintana, F. (2013). A time series model for  responses on the unit interval. \textit{Bayesian Analysis} \textbf{8}, 723--740.

\bibitem[Klein and Moeschberger, 2003]{klein&moeschberger:03}
Klein, J.P. and Moeschberger, M.L. (2003). \textit{Survival Analysis}. Springer: New York.

\bibitem[Lee, 2000]{lee:00}
Lee, R.D. (2000). The Lee-Carter method for forecasting mortality with various extensions and applications. \textit{North American Actuarial Journal} \textbf{4}, 80--91.

\bibitem[Lee and Carter, 1992]{lee&carter:92}
Lee, R.D. and Carter, L.R. (1992). Modeling and forecasting U.S. mortality. \textit{Journal of the American Statistical Association} \textbf{87}, 659--671.

\bibitem[Nieto-Barajas and Walker, 2002]{nieto&walker:02}
Nieto-Barajas, L.E. and Walker, S.G. (2002). Markov beta and gamma process for modelling hazard rates. \textit{Scandinavian Journal of Statistics} \textbf{29}, 413--424.

\bibitem[Plummer, 2003]{plummer:03}
Plummer, M. (2003). JAGS: a program for analysis of Bayesian graphical models using Gibbs sampling.
In \textit{Proceedings of the 3rd International Workshop on Distributed Statistical Computing (DSC 2003)}, March 20--22, Vienna, Austria. ISSN 1609-395X, vol. 124, p. 125.

\bibitem[Renshaw and Haberman, 2003]{renshaw&haberman:03}
Renshaw, A.E. and Haberman, S. (2003). Lee–Carter mortality forecasting with age-specific enhancement. \textit{Insurance: Mathematics and Economics} \textbf{33}, 255--272.

\bibitem[Renshaw and Haberman, 2006]{renshaw&haberman:06}
Renshaw, A.E. and Haberman, S. (2006). A cohort-based extension to the Lee–Carter model for mortality reduction factors. \textit{Insurance: Mathematics and Economics} \textbf{38}, 556--570.

\bibitem[Richman and W\"uthrich, 2019]{richman&wuthrich:19}
Richman, R. and W\"uthrich, M. (2019). A neural network extension of the Lee–Carter model to multiple populations. \textit{Annals of Actuarial Science}, 1-21. doi:10.1017/S1748499519000071

\bibitem[Rossi et al., 2020]{rossi&al:20}
Rossi, R., Murari, A. Gaudio, P. and Gelfusa, M. (2020). Upgrading model selection criteria with goodness of fit tests for practical applications. \textit{Entropy} \textbf{22}, 447.

\bibitem[Smith and Roberts, 1993]{smith&roberts:93}
Smith, A. and Roberts, G. (1993). Bayesian computations via the Gibbs sampler and related Markov chain Monte Carlo methods. {\it Journal of the Royal Statistical Society, Series B} {\bf 55}, 3-–23.

\end{thebibliography}

\newpage 

\begin{table}
\centering
\begin{tabular}{cccccc}
\toprule
 & \multicolumn{2}{c}{$\bep$} & \multicolumn{2}{c}{LC} \\
$N$ & $L_{\rm in}$ & $L_{\rm out}$ & $L_{\rm in}$ & $L_{\rm out}$ & R$_{\rm out}$ \\
\midrule
$1$ & $0.000099$ & $0.000753$ & $0.000387$ & $0.000344$ & $2.18$ \\
$2$ & $0.000102$ & $0.000879$ & $0.000417$ & $0.000499$ & $1.76$ \\
$3$ & $0.000104$ & $0.001068$ & $0.000416$ & $0.000805$ & $1.32$ \\
$4$ & $0.000105$ & $0.001266$ & $0.000445$ & $0.001189$ & $1.06$ \\
$5$ & $0.000108$ & $0.001430$ & $0.000406$ & $0.001419$ & $1.00$ \\
$6$ & $0.000119$ & $0.001593$ & $0.000346$ & $0.001990$ & $0.80$ \\
$7$ & $0.000127$ & $0.001746$ & $0.000269$ & $0.184359$ & $0.01$ \\
$8$ & $0.000153$ & $0.001864$ & $0.000184$ & $0.735184$ & $0.00$ \\
$9$ & $0.000167$ & $0.002124$ & $0.000177$ & $1.214812$ & $0.00$ \\
$10$ & $0.000169$ & $0.002413$ & $0.000160$ & $8.422003$ & $0.00$ \\
\bottomrule
\end{tabular}
\caption{L-measures comparison at $\nu=1/2$ for $N=1,2,\ldots,10$ times left out. In sample ($L_{\rm in}$) and out of sample ($L_{\rm out}$), for our best fitting model $\bep(p=1,q=10,c=20)$ and LC model. Last column shows the ratio between the out of sample L-measures.}
\label{tab:gof}
\end{table}

\begin{figure}
\centerline{\includegraphics[scale=0.7]{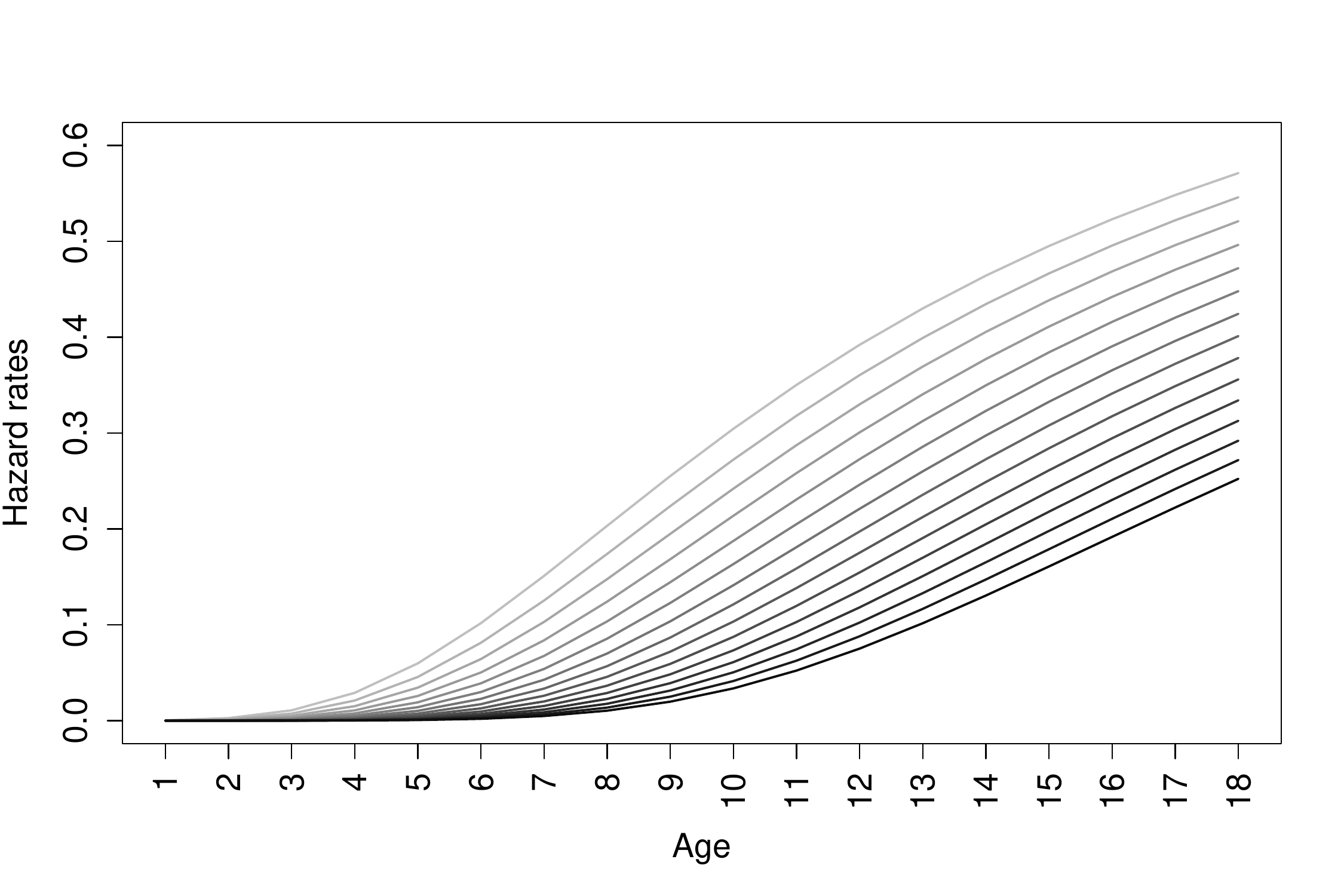}}
\caption{{\small True hazards $h_t(x)$ in simulated study for $t=1,\ldots,15$. Darker colour corresponds to larger $t$.}}
\label{fig:trueh}
\end{figure}

\begin{figure}
\centerline{\includegraphics[scale=0.7]{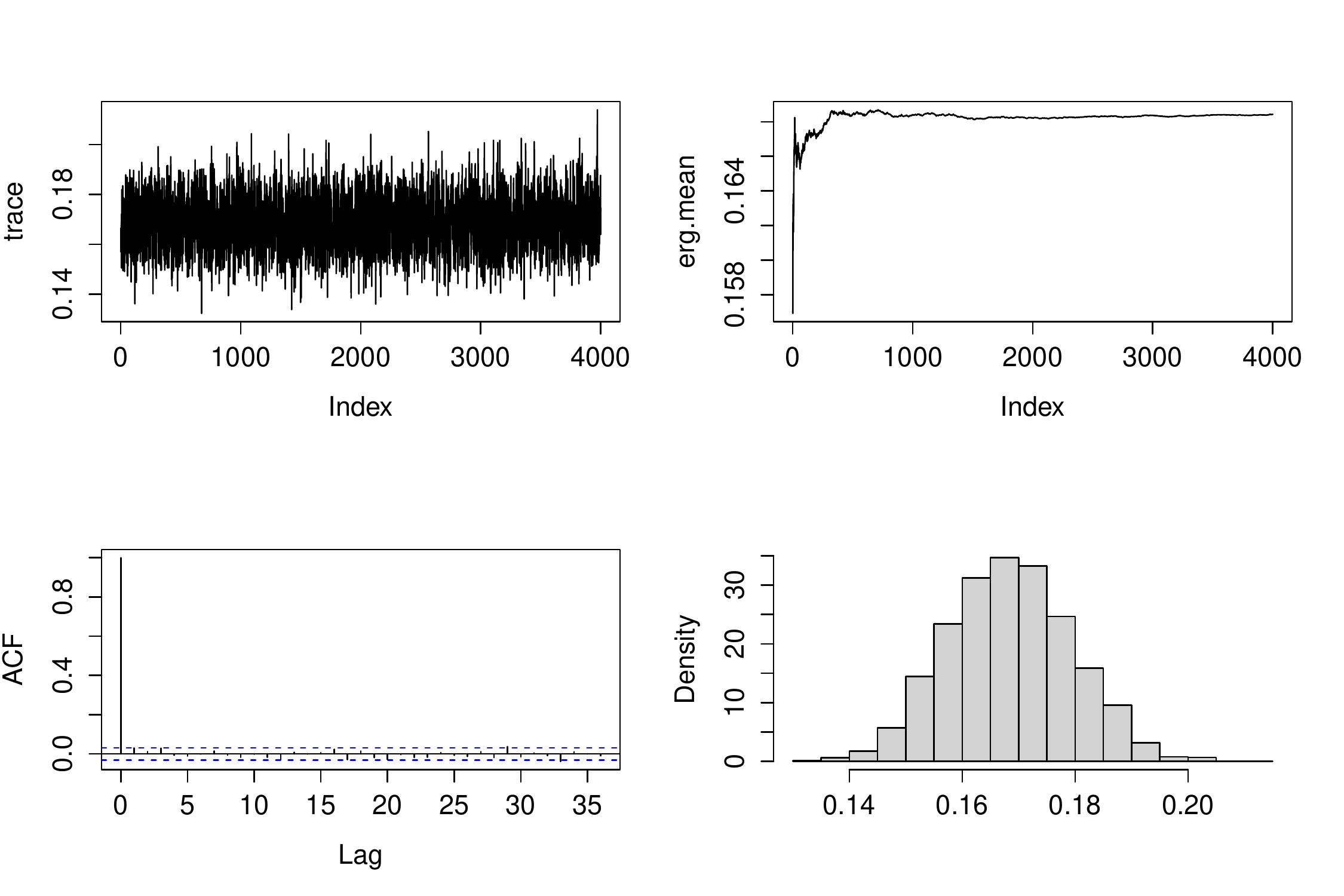}}
\caption{{\small MCMC convergence diagnostics for $\omega$ when $p=10$, $q=8$ and $c=5$. Trace plot, ergodic means, autocorrelation function and probability histogram.}}
\label{fig:mcmc}
\end{figure}

\begin{figure}
\centerline{\includegraphics[scale=0.6]{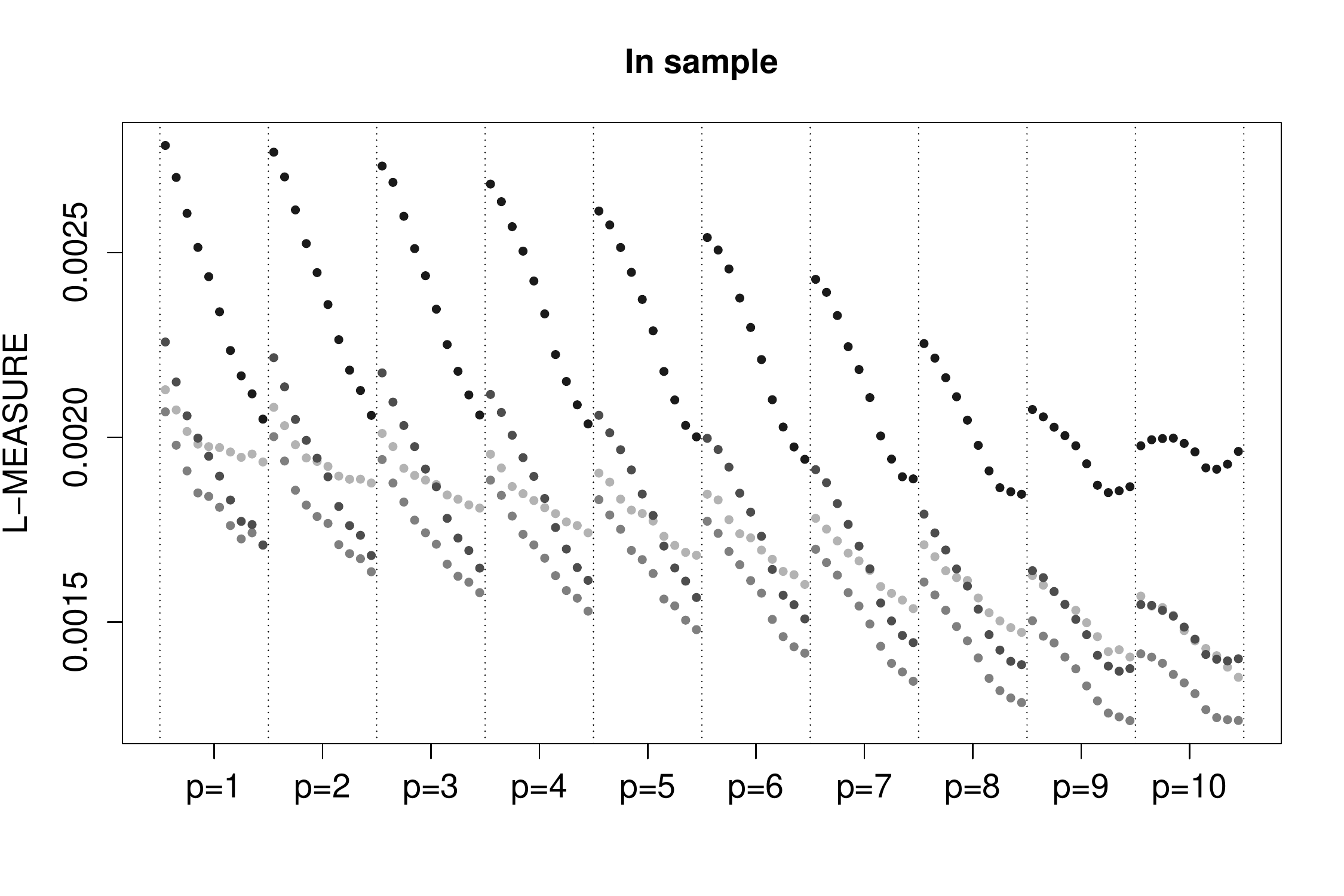}}
\centerline{\includegraphics[scale=0.6]{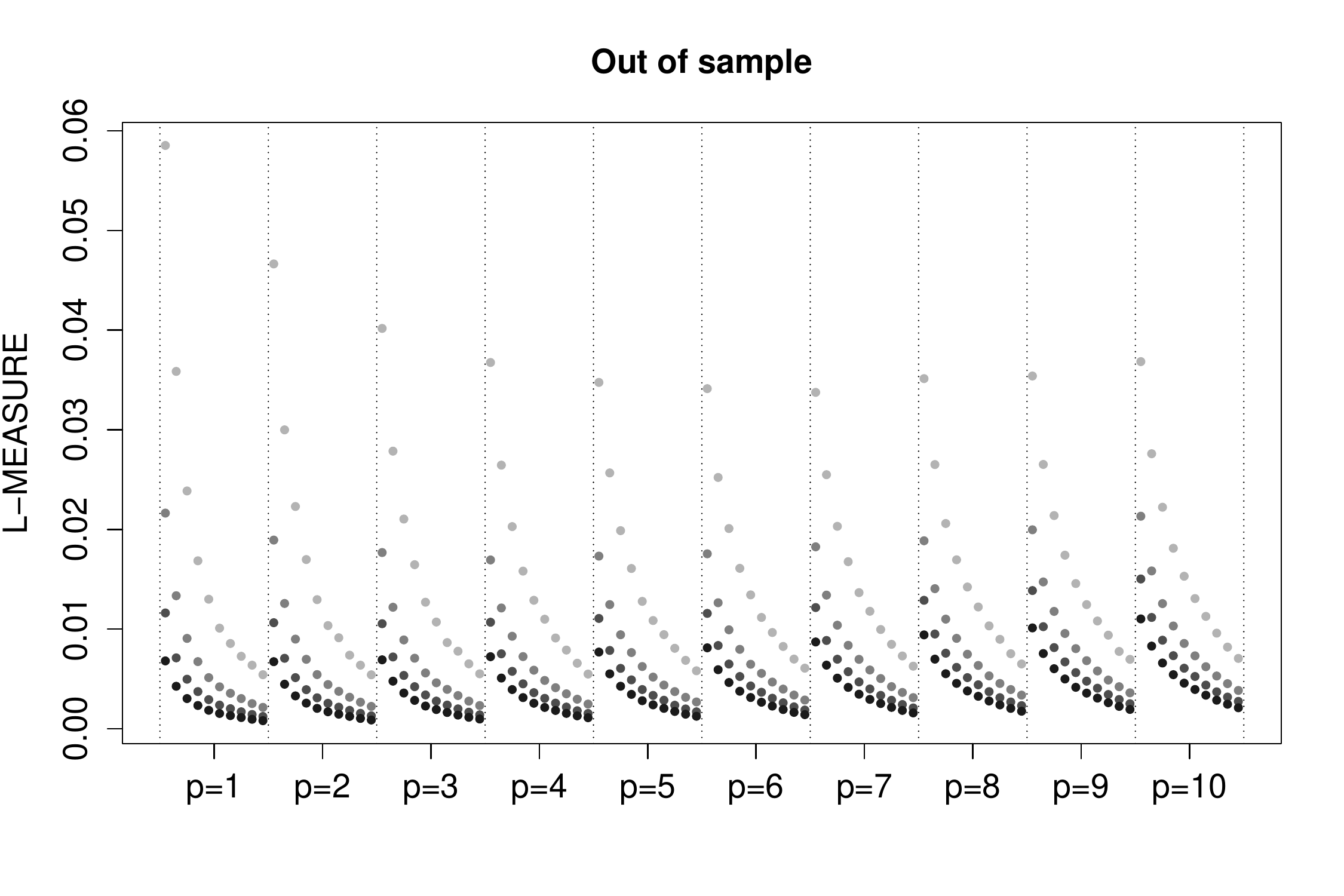}}
\caption{{\small L-measure gof statistics for different prior specifications. In sample (top), out of sample (bottom). $p=1,\ldots,10$; $q=1,\ldots,10$ (within each block); and $c=2,5,10,20$ in colour intensity from light to dark.}}
\label{fig:lmeaS}
\end{figure}

\begin{figure}
\centerline{\includegraphics[scale=0.35]{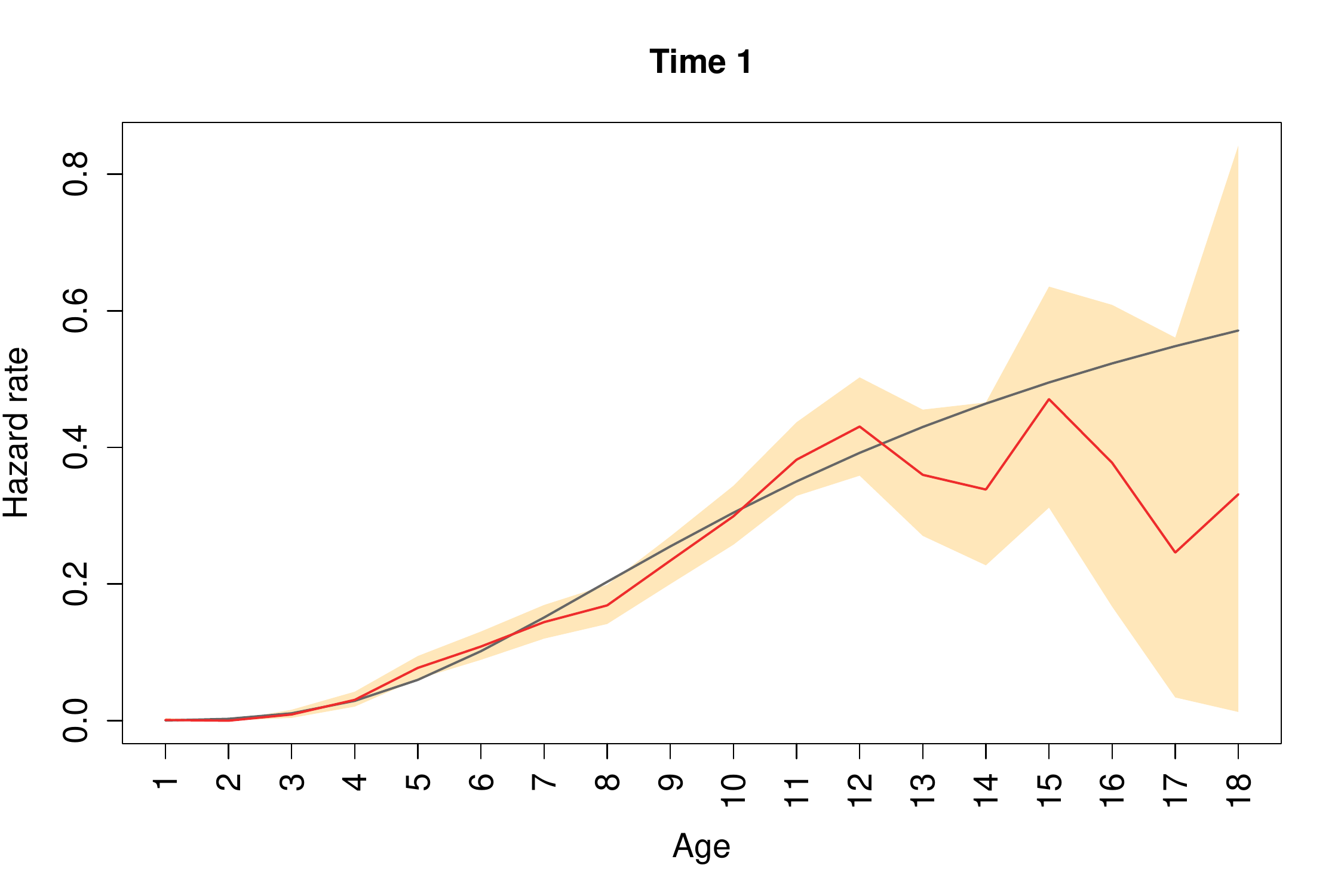}
\includegraphics[scale=0.35]{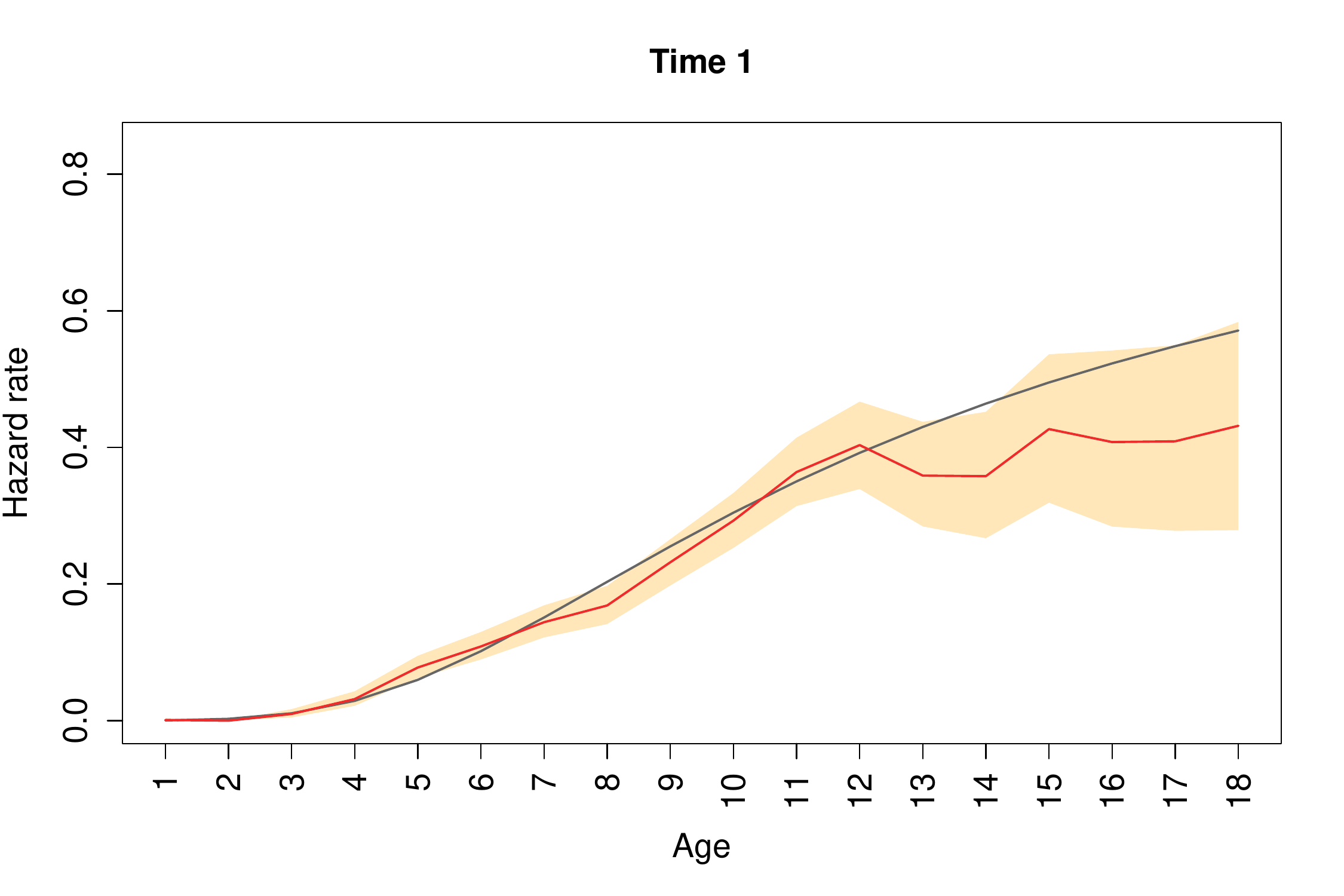}}
\centerline{\includegraphics[scale=0.35]{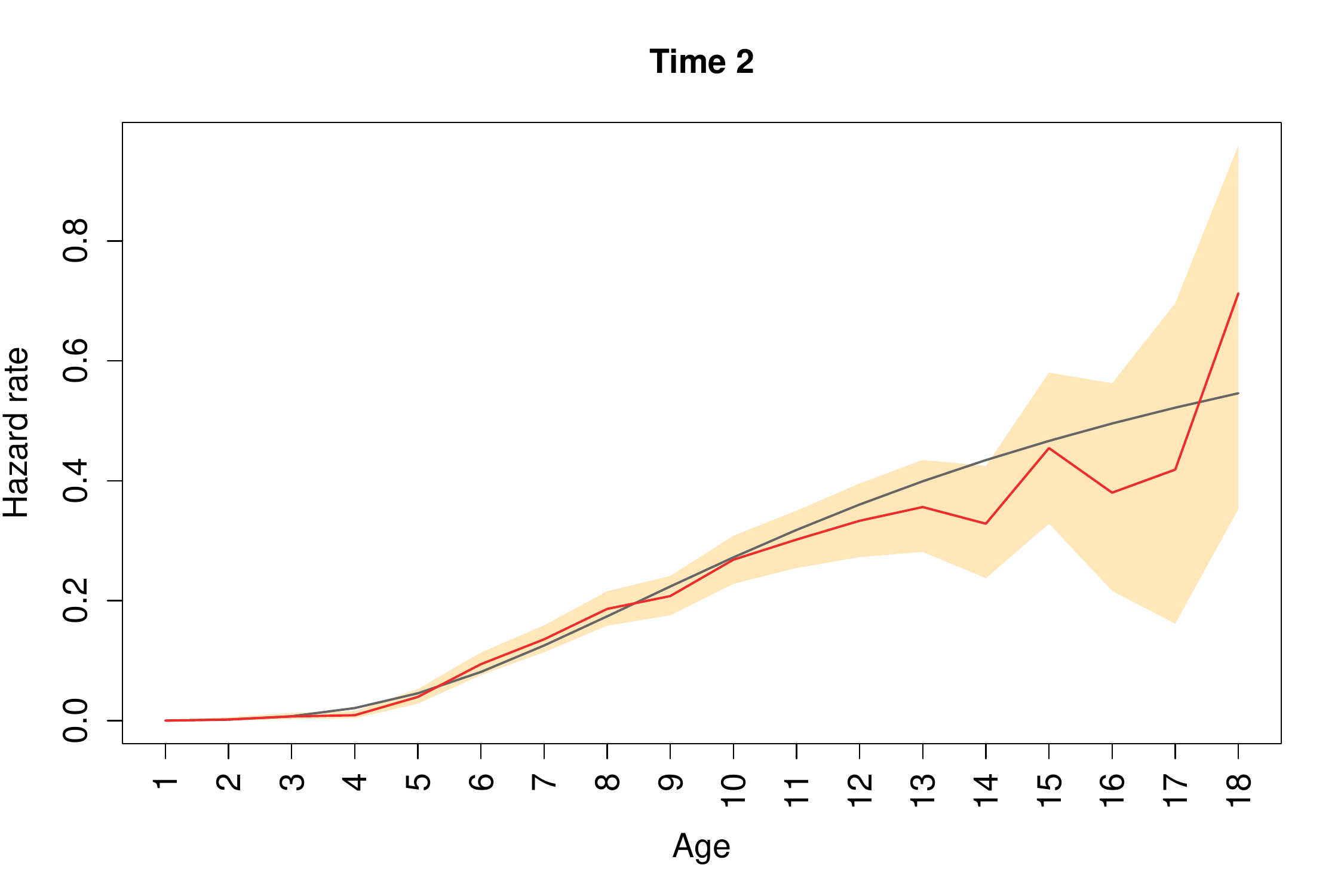}
\includegraphics[scale=0.35]{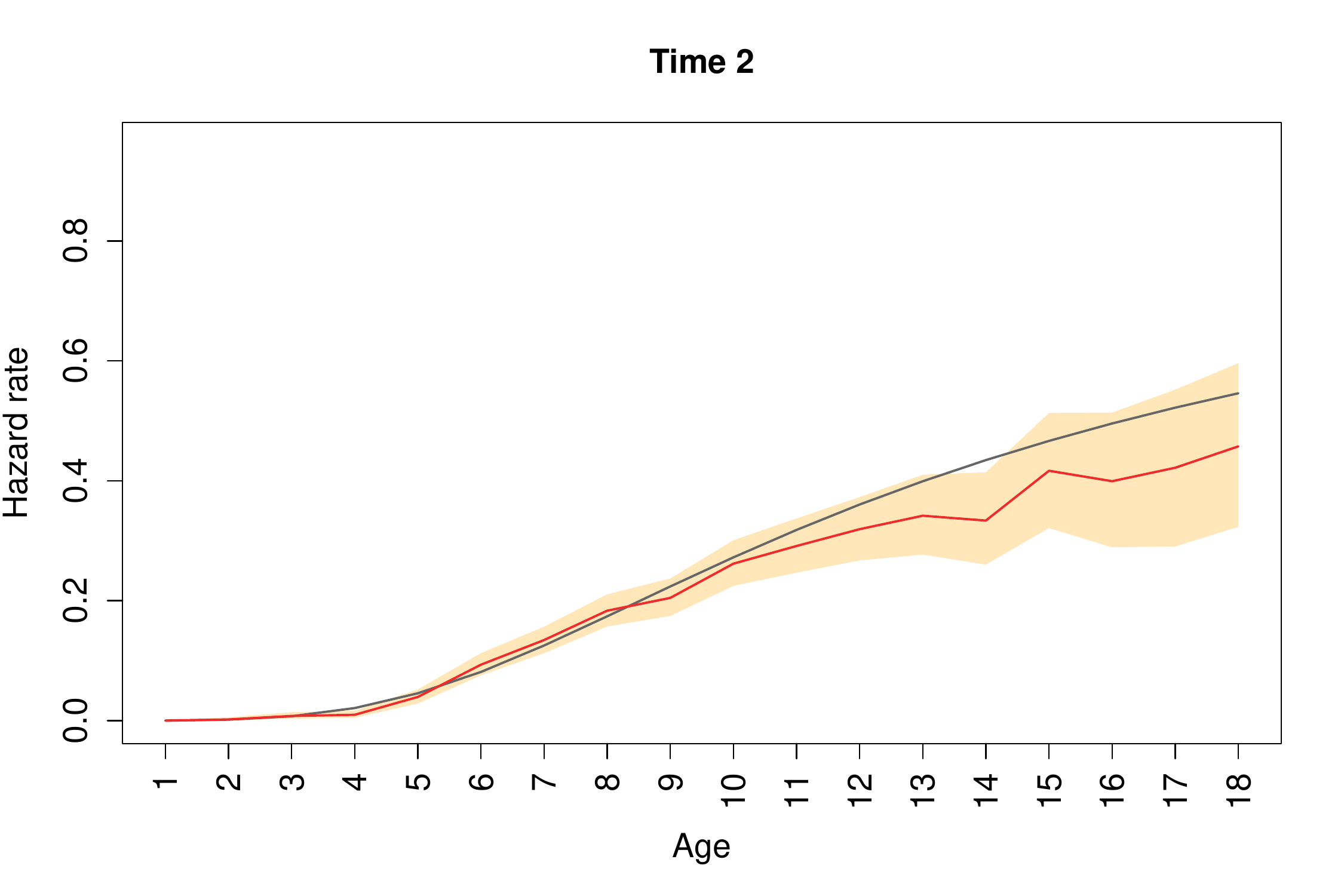}}
\centerline{\includegraphics[scale=0.35]{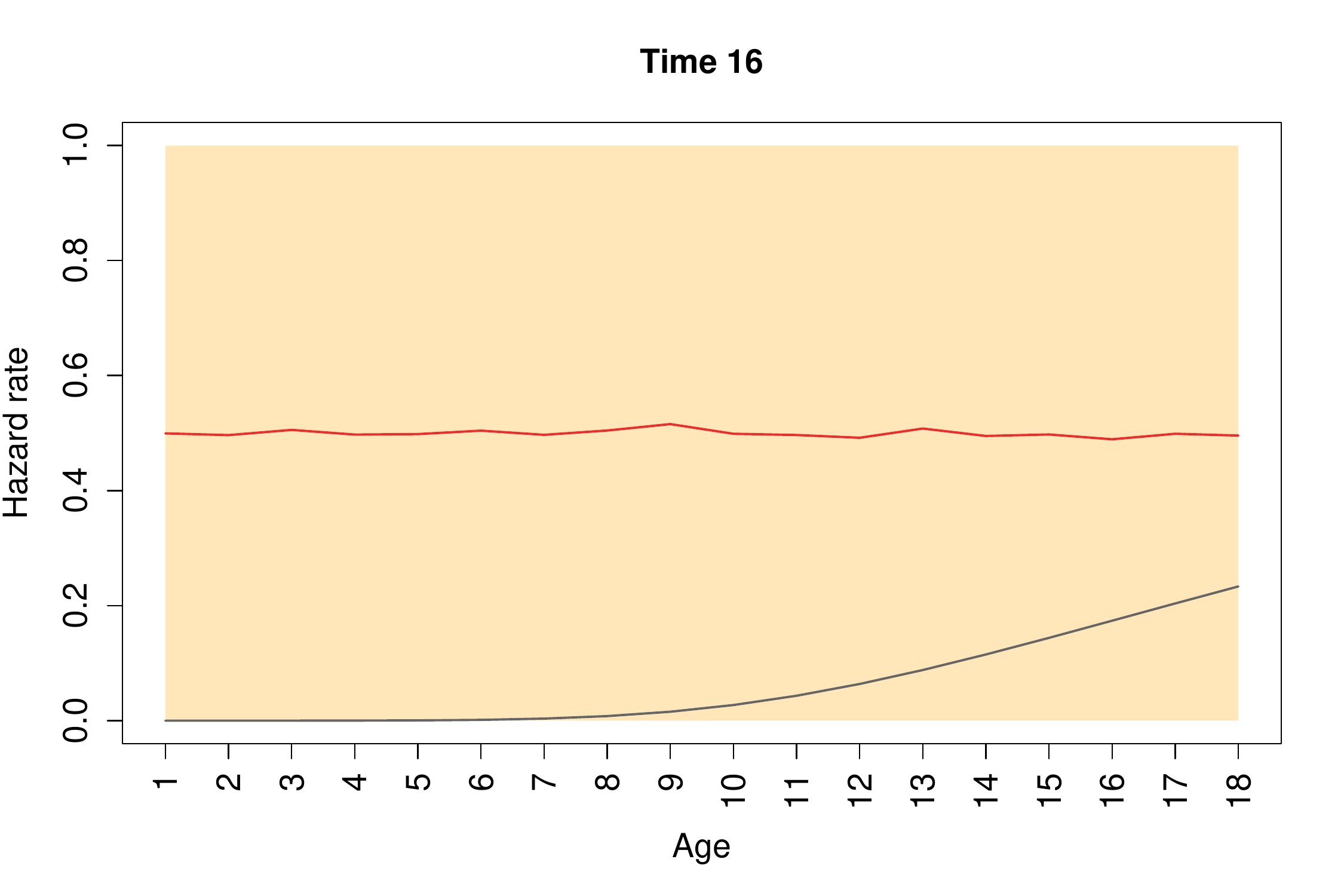}
\includegraphics[scale=0.35]{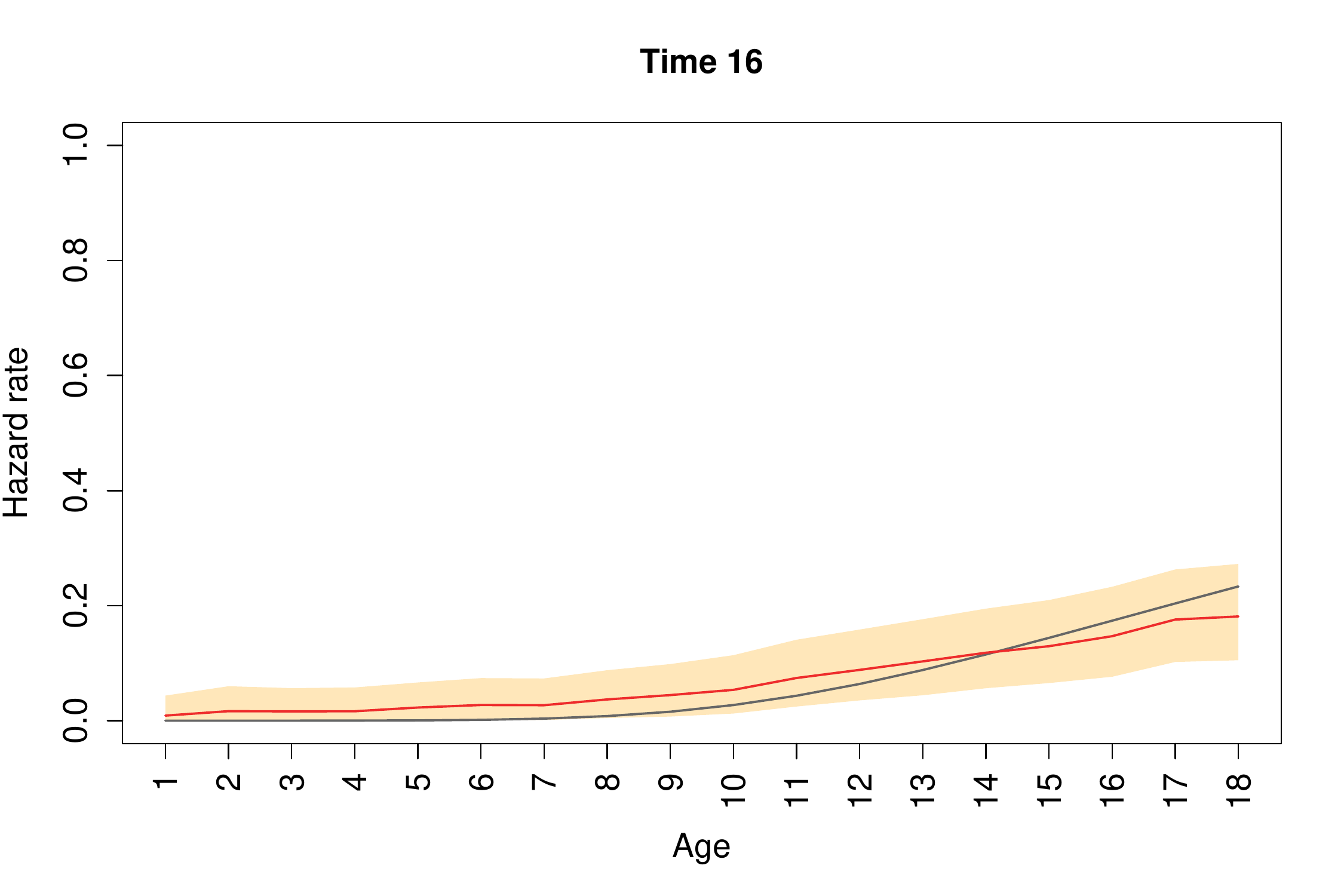}}
\centerline{\includegraphics[scale=0.35]{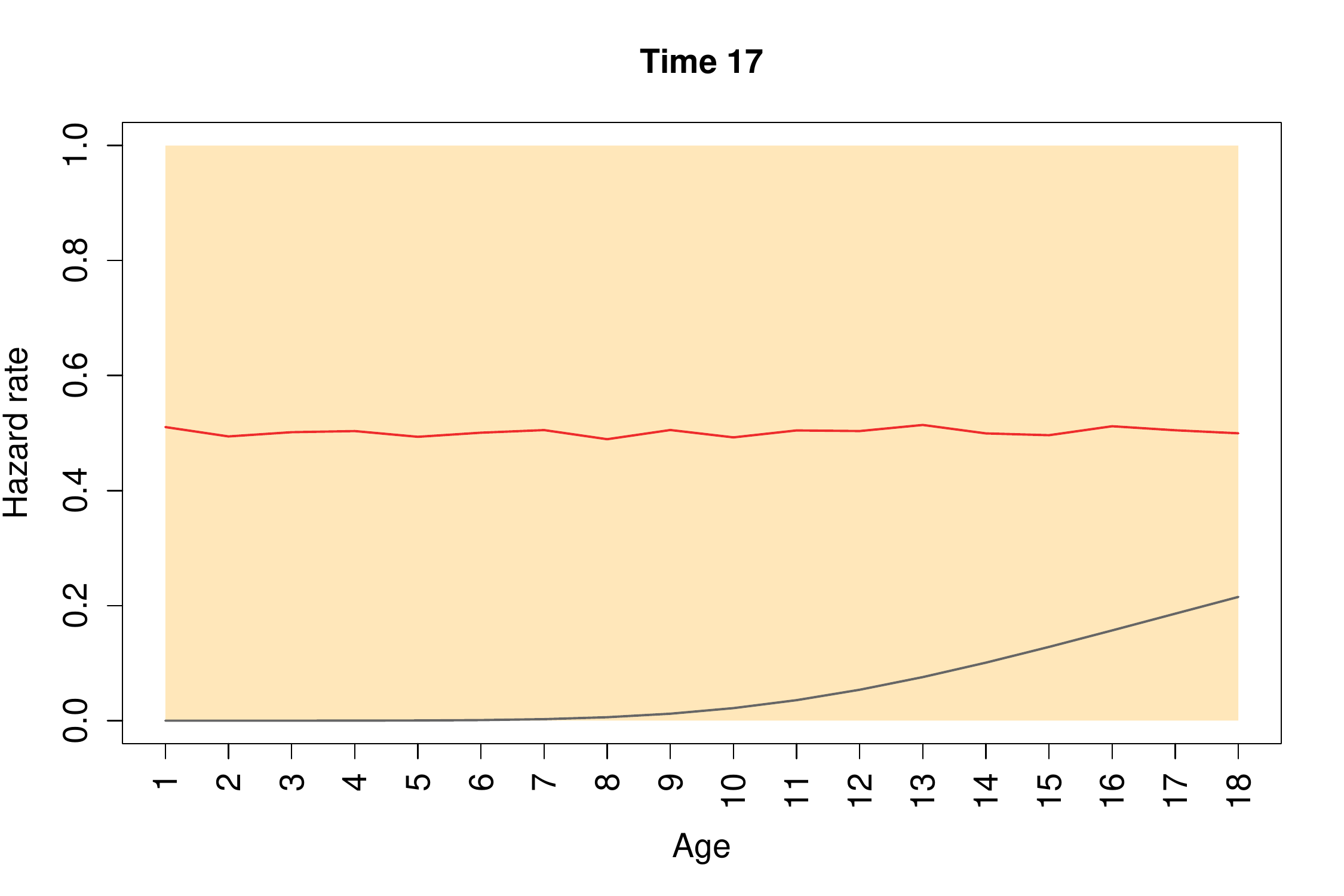}
\includegraphics[scale=0.35]{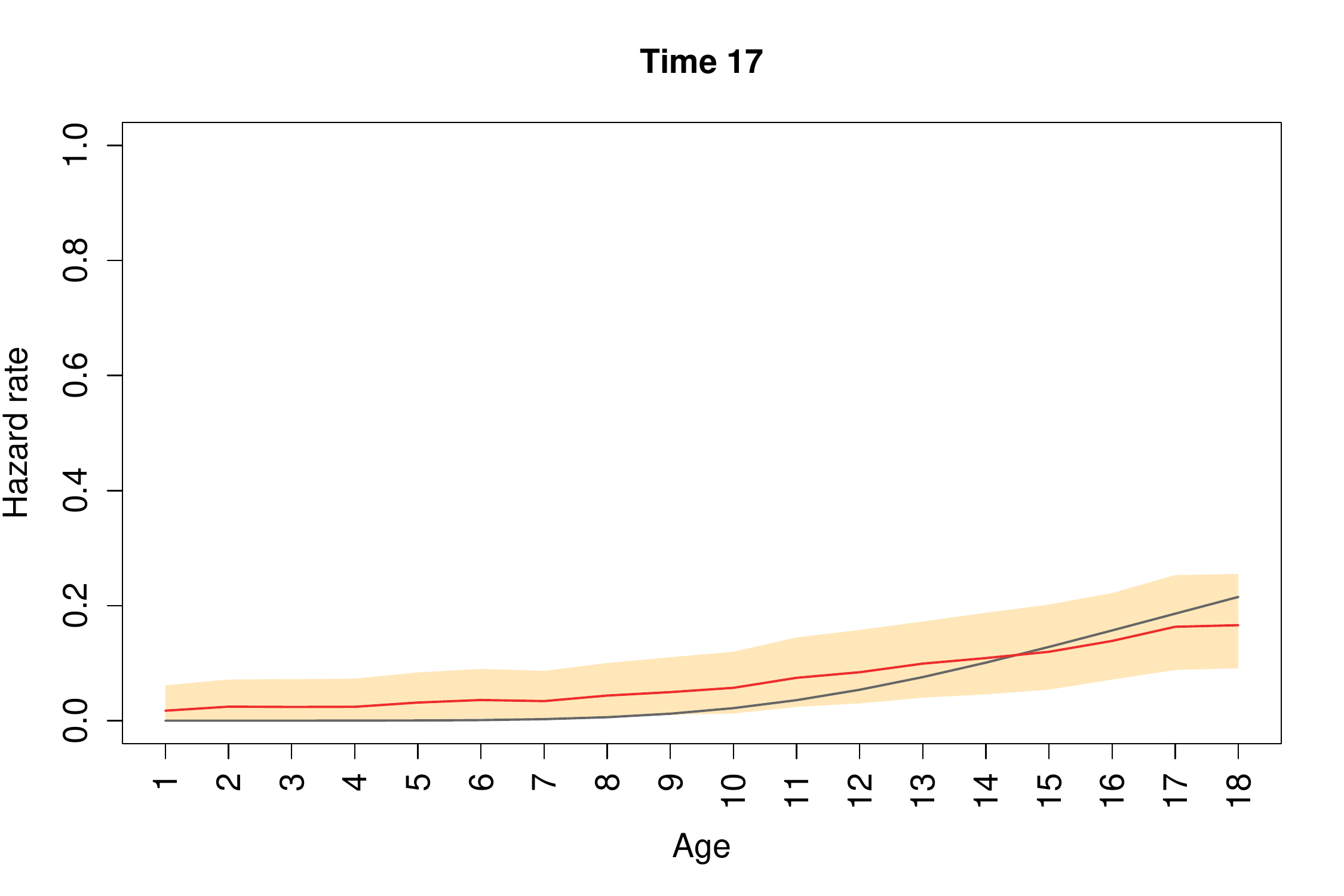}}
\caption{{\small Posterior estimates of hazard rates $h_t(x)$ for $t=1,2,16,17$. Independent prior with $c=0$ (left column), dependent prior (right column) with $p=10$, $q=8$ and $c=5$ (in sample) and with $p=1$, $q=10$ and $c=10$ (out of sample).}}
\label{fig:posth}
\end{figure}

\begin{figure}
\centerline{\includegraphics[scale=0.7]{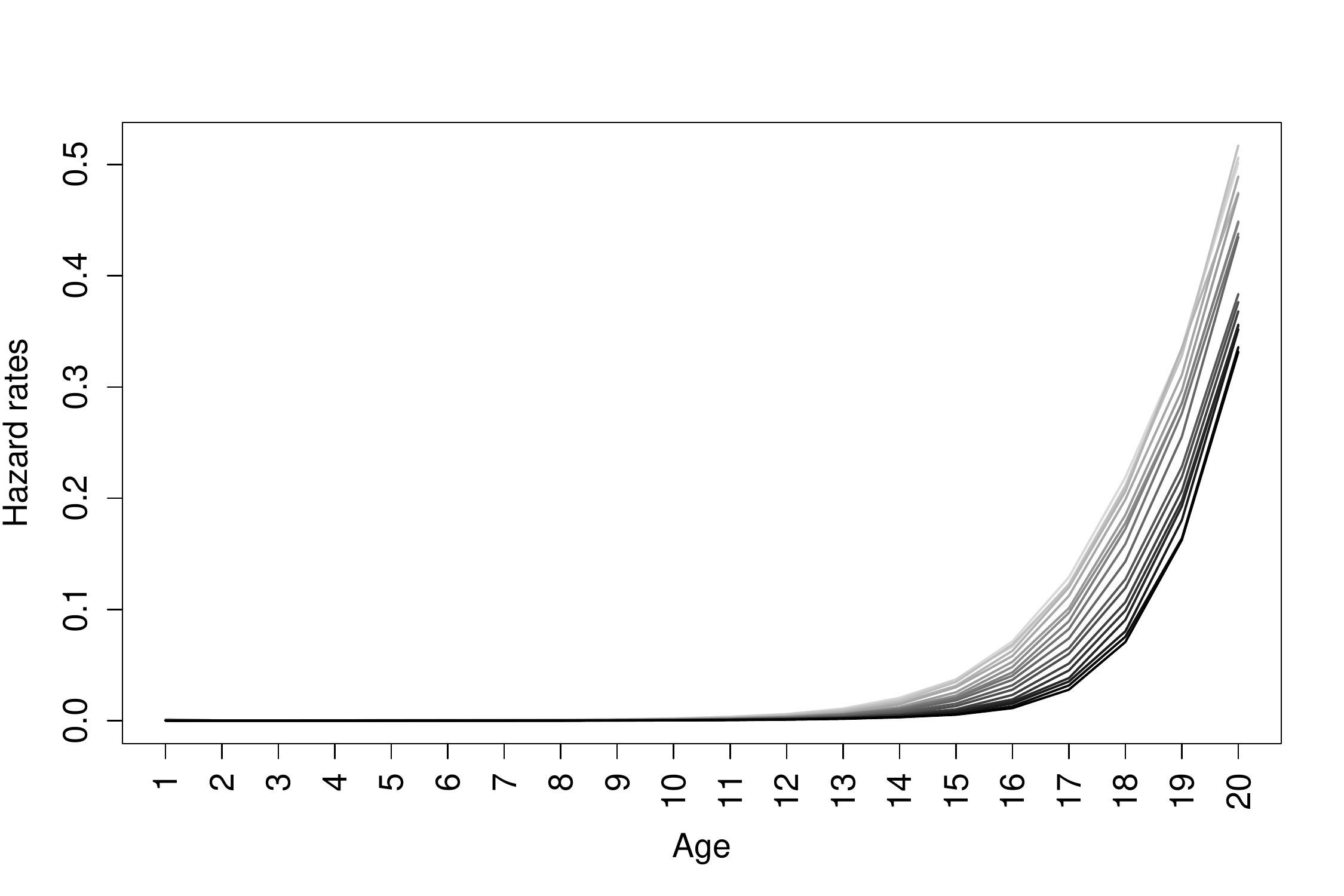}}
\caption{{\small Frequentist nonparametric estimates of hazard rates $h_t(x)$ for $t=1,\ldots,18$. Darker colour corresponds to larger $t$.}}
\label{fig:swiss0}
\end{figure}

\begin{figure}
\centerline{\includegraphics[scale=0.6]{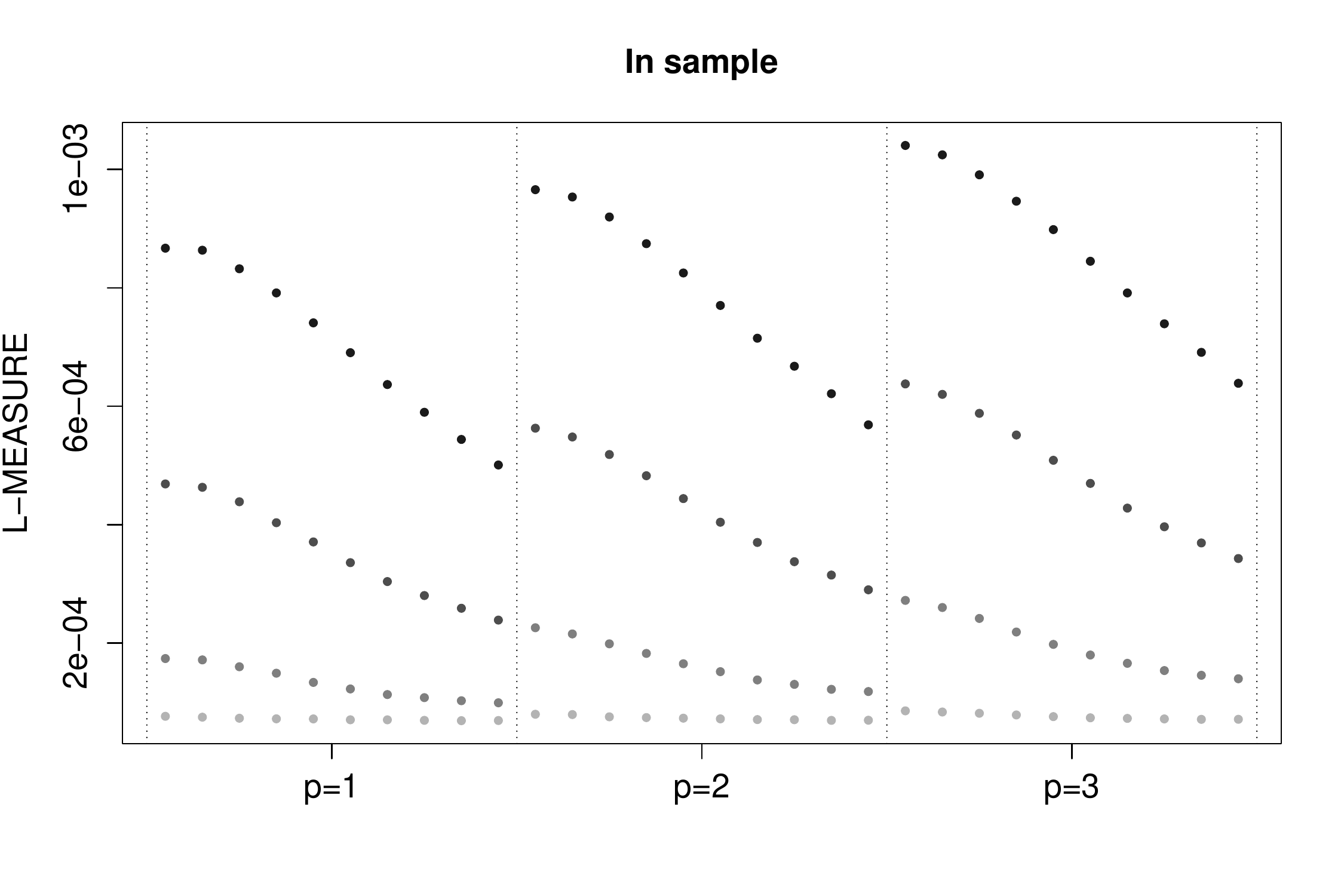}}
\centerline{\includegraphics[scale=0.6]{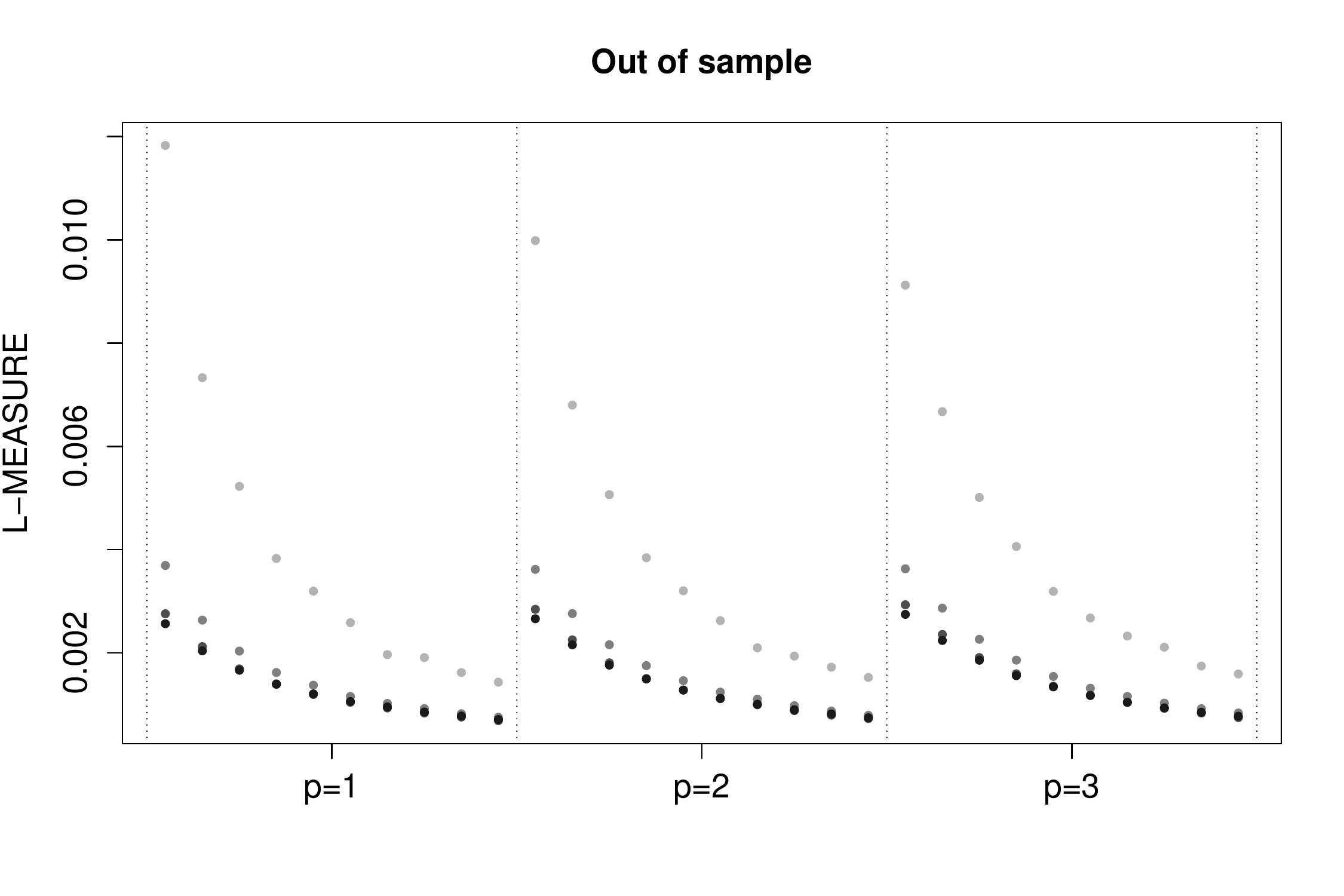}}
\caption{{\small L-measure gof statistics for different prior specifications. In sample (top), out of sample (bottom). $p=1,\ldots,3$; $q=1,\ldots,10$ (within each block); and $c=5,20,50,100$ (in colour intensity from light to dark).}}
\label{fig:lmeaR}
\end{figure}

\begin{figure}
\centerline{\includegraphics[scale=0.35]{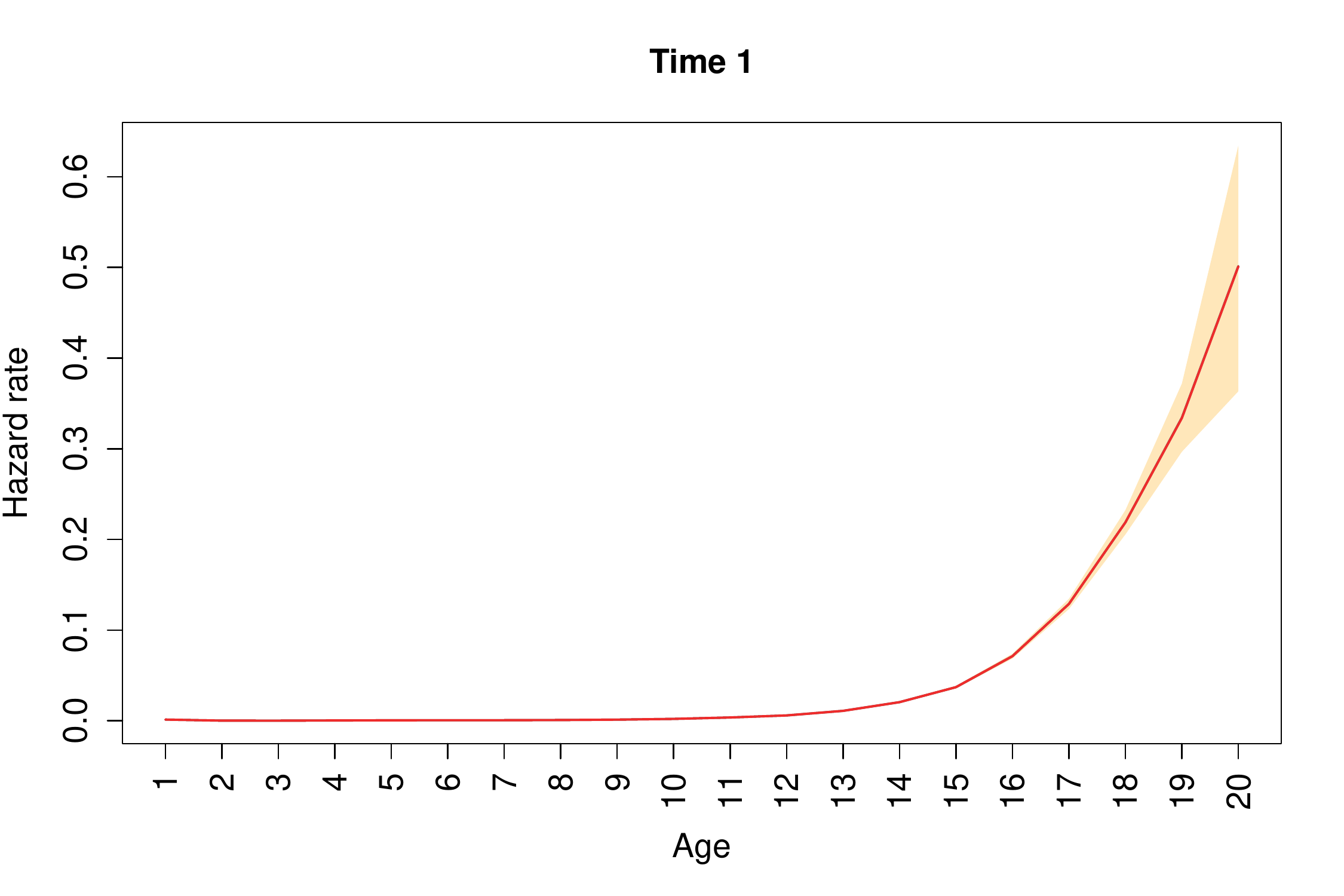}
\includegraphics[scale=0.35]{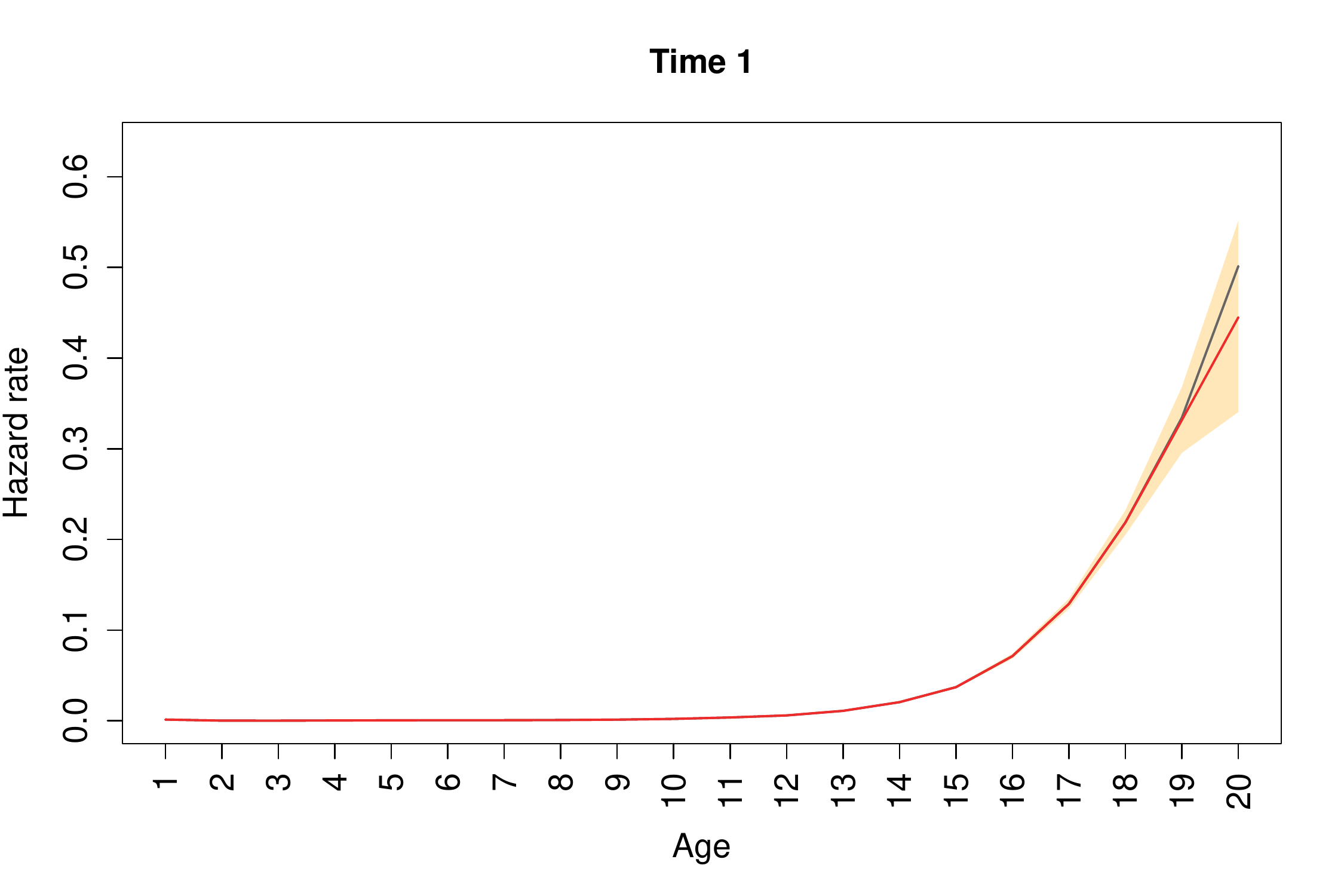}}
\centerline{\includegraphics[scale=0.35]{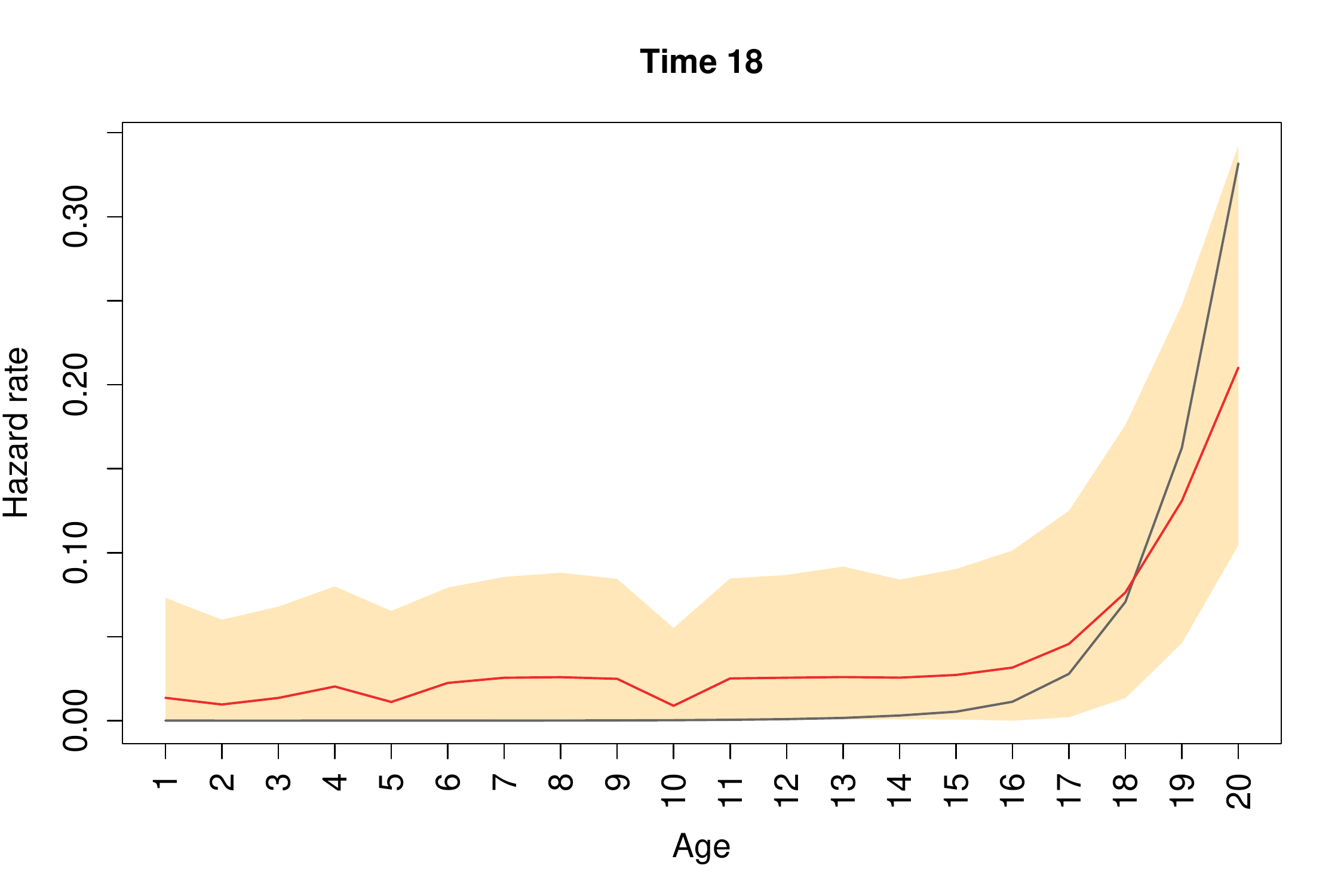}
\includegraphics[scale=0.35]{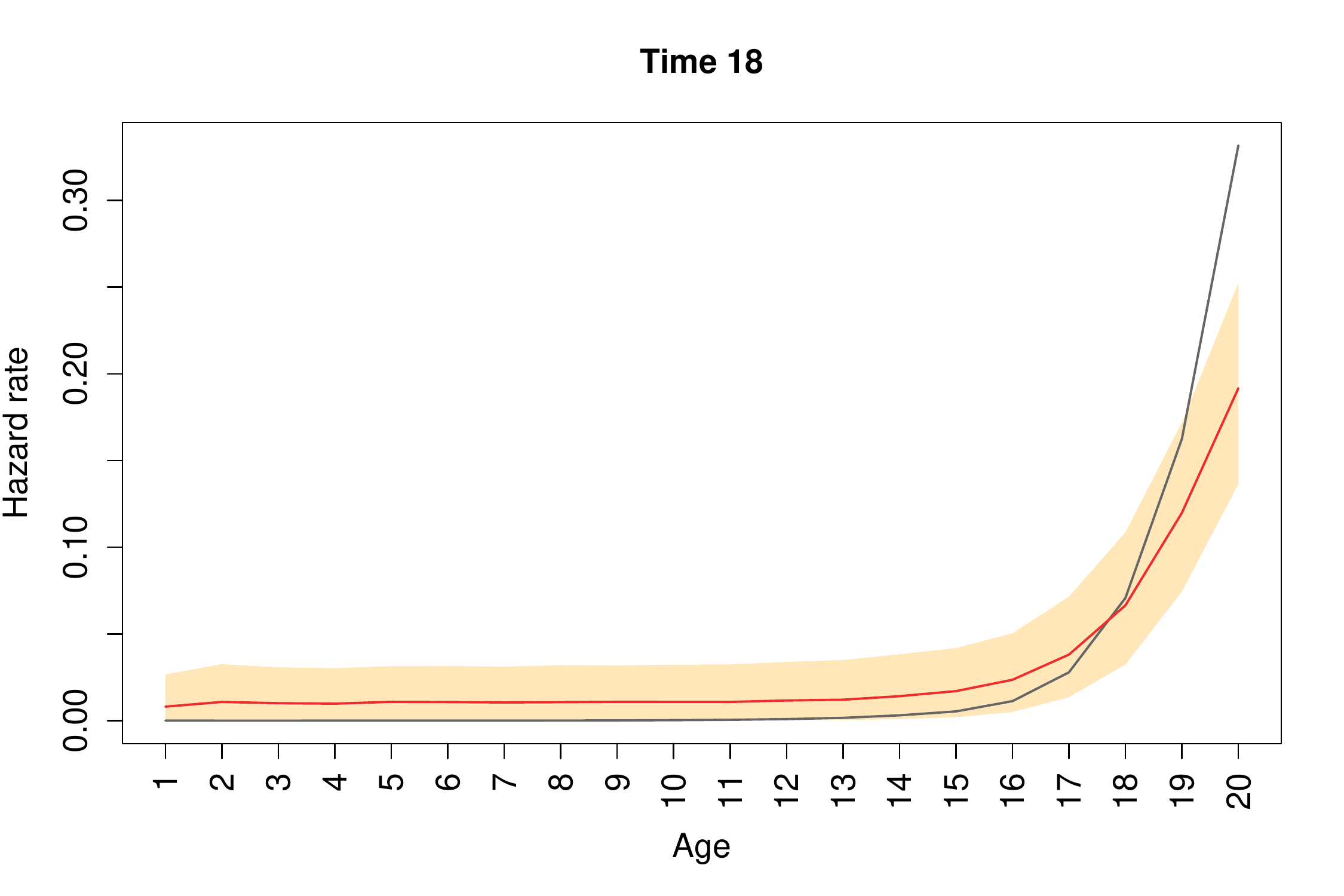}}
\caption{{\small Posterior estimates of hazard rates $h_t(x)$ for in sample (top row) and out of sample (bottom row). For $t=1$ with $c=0$  (top left); for $t=1$ with $p=1$, $q=10$ and $c=20$ (top right). For $t=18$ with $p=1$, $q=10$ and $c=5$ (bottom left); for $t=18$ with $p=1$, $q=10$ and $c=20$ (bottom right).}}
\label{fig:swiss1}
\end{figure}

\end{document}